\documentclass[reprint, amsmath, amssymb, aps,superscriptaddress,pra,nofootinbib]{revtex4-1}

\usepackage{hyperref}
\hypersetup{colorlinks=true,linkcolor=blue,citecolor = blue, urlcolor=blue}

\usepackage{graphicx}
\usepackage{dcolumn}
\usepackage{bm}
\usepackage[utf8]{inputenc}
\usepackage{amssymb}
\usepackage{amsthm}
\usepackage{mathrsfs}
\usepackage{dsfont}
\usepackage{multirow}
\usepackage{enumitem}
\usepackage{ stmaryrd } 
\usepackage{float}
\usepackage{physics}
\usepackage{cancel}
\usepackage{array}
\usepackage{longtable}
\usepackage{soul}
\usepackage{breakurl}

\usepackage{orcidlink}
\usepackage{tikzit}

\tikzstyle{red dot}=[fill=red, draw=none, shape=circle]
\tikzstyle{green dot}=[fill={rgb,255: red,70; green,176; blue,20}, draw=none, shape=circle]
\tikzstyle{Resize}=[font={\scriptsize}]
\tikzstyle{Big}=[font={\huge}]
\tikzstyle{Photon}=[fill={rgb,255: red,255; green,230; blue,103}, draw=none, shape=circle]
\tikzstyle{Big node}=[fill={rgb,255: red,191; green,191; blue,191}, draw=black, shape=circle]
\tikzstyle{small node}=[fill={rgb,255: red,191; green,191; blue,191}, draw=black, shape=circle, inner sep=0, minimum size=3pt]
\tikzstyle{Blue dot}=[fill=blue, draw=black, shape=circle]

\tikzstyle{Fill red}=[-, fill={rgb,255: red,255; green,184; blue,184}]
\tikzstyle{dashes}=[-, dashed, dash pattern=on 1mm off 1mm]
\tikzstyle{Fill grey}=[-, fill={rgb,255: red,166; green,166; blue,166}]
\tikzstyle{Thick}=[-, thick]
\tikzstyle{Arrow}=[->]
\tikzstyle{red line}=[-, fill=none, draw=red]
\tikzstyle{Blue line}=[-, fill=none, draw=blue]
\tikzstyle{Gray line}=[-, fill=none, draw=gray]
\tikzstyle{Thick arrow}=[thick, ->]
\tikzstyle{Fill pink}=[-, fill={rgb,255: red,255; green,140; blue,140}]
\tikzstyle{dashed arrow}=[->, dashed, dash pattern=on 1mm off 0.5mm]
\tikzstyle{arrow}=[->, very thick, draw=red]
\tikzstyle{dashes red}=[-, dashed, dash pattern=on 1mm off 1mm, fill=none, draw=red]
\tikzstyle{dashes blue}=[-, dashed, dash pattern=on 1mm off 1mm, fill=none, draw=blue]
\tikzstyle{Thick dashed arrow}=[->, thick, dashed, dash pattern=on 1mm off 0.5mm]
\tikzstyle{Fill green}=[-, fill={rgb,255: red,70; green,176; blue,20}]
\tikzstyle{Fill real red}=[-, fill=red]

\newcommand{\N}{\mathbb{N}}
\newcommand{\R}{\mathbb{R}}

\newcommand{\C}{\mathbb{C}}

\newcommand{\1}{\mathds{1}}

\newcommand{\mycomment}[1]{}

\newtheorem{thm}{Theorem}

\begin{document}

\preprint{APS/123-QED}

\title{Measuring entanglement along collective operators}

\author{\'Eloi Descamps\orcidlink{0000-0002-6911-452X}}
\email{eloi.descamps@u-paris.fr}
\affiliation{Universitée Paris Cité, CNRS, Laboratoire Matériaux et Phénomènes Quantiques, 75013 Paris, France}
\author{Arne Keller\orcidlink{0000-0002-6934-7198}}
\affiliation{Universitée Paris Cité, CNRS, Laboratoire Matériaux et Phénomènes Quantiques, 75013 Paris, France}
\affiliation{Département de Physique, Université Paris-Saclay, 91405 Orsay Cedex, France}
\author{Pérola Milman\orcidlink{0000-0002-7579-7742}}
\affiliation{Universitée Paris Cité, CNRS, Laboratoire Matériaux et Phénomènes Quantiques, 75013 Paris, France}

\date{\today}

\begin{abstract}
We introduce a framework for studying multiparty entanglement by analyzing the behavior of collective variables. Throughout the manuscript, we explore a specific type of multiparty entanglement that can be detected through the fluctuations of a collective observable. We thoroughly analyze its properties and how it can be extended to mixed states while placing it within the context of the existing literature. The novelty of our approach also lies in presenting a graphical point of view. This is done by introducing a spectral space on which the various properties of our entanglement quantifier have a direct pictorial interpretation. Notably, this approach proves particularly effective for assessing $k$-entanglement, as we show its ability to extend previously established inequalities. To enhance understanding, we also demonstrate how this framework applies to specific scenarios, encompassing both finite-dimensional cases and infinite-dimensional systems, the latter being exemplified by the time-frequency modal degree of freedom of co-propagating single photons.
\end{abstract}

\maketitle


\section{Introduction}\label{section: intro}
Entanglement stands as a cornerstone of quantum systems, offering a distinct advantage over classical counterparts in various applications such as sensing \cite{hyllus_fisher_2012, toth_quantum_2014,pezze_quantum_2014}, computing \cite{nielsen_quantum_2010,menicucci_universal_2006,shor_polynomial-time_1997}, and communication \cite{bennett_quantum_2014,gisin_quantum_2002}. Despite extensive studies, the exploration of entanglement remains an ongoing endeavor, particularly in the realm of multipartite entanglement quantification. While bipartite entanglement is well-understood \cite{horodecki_separability_1996,sperling_necessary_2009}, the treatment of multipartite entanglement, even for pure states, lacks consensus, with various perspectives depending on different system partitions \cite{szalay_multipartite_2015} or a global standpoint \cite{hyllus_fisher_2012,toth_multipartite_2012} (see \cite{horodecki_quantum_2009,amico_entanglement_2008} for reviews). In this context, searching for quantifiers of entanglement \cite{plenio_introduction_2006,ma_multipartite_2024} and entanglement witnesses \cite{li_entanglement_2013,laurell_witnessing_2024,akbari-kourbolagh_entanglement_2019} allows for a better understanding of the multiple facets of entanglement.

In \cite{guhne_multipartite_2005}, it was proposed that the collective spin operators' variance serves as a meaningful numerical quantity indicative of multipartite entanglement in the context of spin systems. This observation was later generalized to a broader class of systems \cite{pezze_quantum_2014,chen_wigner-yanase_2005,hong_detection_2021}, albeit limited to finite-dimensional ones, or with complex generalized constructions \cite{hong_detecting_2015}. This paper aims to present a unified framework that extends and generalizes these previous findings. We introduce a novel quantifier inspired by prior works and explore its extension to mixed states.

A specific facet of multipartite entanglement is $k$-entanglement \cite{guhne_entanglement_2009}, where the introduced quantifier finds particular relevance. Building upon the arguments and derivations of \cite{hyllus_fisher_2012,chen_wigner-yanase_2005}, we establish that our quantifier can also furnish information, similar to a known inequality, whose violation serves as a hallmark of $k$-entanglement. We present a more general theorem and discuss variations of this inequality, accounting for scenarios where prior knowledge of collective entanglement is available.

To elucidate the connection between our framework and existing literature, a significant portion of the paper is dedicated to illustrating how our general approach can be applied to concrete situations. We delve into two specific scenarios: finite-dimensional Hilbert spaces and the realm of time-frequency single photon states \cite{fabre_time_2022}.

The paper unfolds as follows: in Section \ref{sec: system studied}, we present the general context as well as the different examples used as illustrations throughout the paper. In Section \ref{sec: entanglement measure}, we introduce the fundamentals of the variance of a collective operator and demonstrate how it can be used to introduce a new entanglement witness. Finally, in Section \ref{sec: k sep}, we discuss how this entanglement witness is particularly relevant to the study of $k$-entanglement and provide a statement that generalizes the existing literature.

\section{Systems studied}\label{sec: system studied}
We consider $\mathcal H$ to be a Hilbert space (finite or infinite-dimensional) and we fix $\hat H$ to be any Hamiltonian acting on $\mathcal H$. For simplicity, we will assume throughout the paper that the eigenvalues of $\hat H$ are non-degenerate. For an integer $n$, we consider the $n$-partite space $\mathcal H^{\otimes n}$. We call $\hat H_i$ the operator obtained from $\hat H$ that acts on the $i$-th party: $\hat H_i=\hat{\1}^{\otimes (i-1)}\otimes \hat H \otimes \hat{\1}^{\otimes (n-i)}$. We can finally introduce 
\begin{equation}
    \hat H_\text{coll}=\sum_{i=1}^n \hat H_i,
\end{equation}
the collective operator acting on the whole system. Now, we discuss different settings that fit into this picture. Throughout the paper, we will refer to these examples for illustration purposes. Moreover, depending on the situation, the specific properties of these examples will be used to allow further investigations.

\subsection{Finite dimensional setting}\label{subsec: finite dim setting}
The first example we will describe is the case of a finite-dimensional $\mathcal H$. This situation implies that the operator $\hat H$ has only a finite number of eigenvalues, meaning in particular that the spectrum of $\hat H$ is bounded. Thus, in the following, $\lambda_{\max}$ and $\lambda_{\min}$ will respectively designate the maximal and minimal eigenvalues of $\hat H$.\\

A quintessential example of this scenario is when $\mathcal{H} = \mathbb{C}^2$ and $\hat{H} = \hat{Z}$, where $\hat{Z}$ is the Pauli-$Z$ matrix. This corresponds to the spin-$\frac{1}{2}$ case. In this context, $\hat{H}_\text{coll} = \sum_{i=1}^n \hat{Z}_i$ represents the sum of spin operators. This situation encapsulates the study of $n$ spin-$\frac{1}{2}$ particles, which is the framework where the concepts presented in this paper were originally developed \cite{guhne_multipartite_2005}.

\subsection{Time-frequency case}\label{subsec: time-frequency case}
 Opposite to finite-dimensional systems are infinite-dimensional ones. We will thus develop an example of one such situation. We consider the space of single-photon time-frequency states $\mathcal S$. We recall that a single-photon time-frequency state $\ket{\psi}$ is characterized by its normalized spectrum $S$ such that 
\begin{equation}
    \ket{\psi}=\int \dd \omega S(\omega)\hat a^\dagger(\omega)\ket{\text{vac}},
\end{equation}
with $S(\omega)\in\C$ and $\int \dd\omega \abs{S(\omega)}^2=1$. In this equation, $\hat a^\dagger(\omega)$ is the operator that creates a photon at frequency $\omega$. It satisfies the usual bosonic commutation relation $[\hat a(\omega),\hat a(\omega')]=0$ and $[\hat a(\omega),\hat a^\dagger(\omega')]=\delta(\omega-\omega')\hat{\1}$. $\ket{\text{vac}}$ is the vacuum state. Following Ref.~\cite{fabre_time_2022}, we define the frequency operator $\hat\omega$ as
\begin{equation}
    \hat\omega=\int \dd \omega \omega \hat a^\dagger(\omega)\hat a(\omega),
\end{equation}
whose action on the state $\ket{\psi}$ is given by $\hat \omega\ket{\psi}=\int \dd\omega \omega S(\omega)\hat a^\dagger(\omega)\ket{\text{vac}}$. As a shorthand, we also introduce the notation $\ket{\omega}=\hat a^\dagger(\omega)\ket{\text{vac}}$. The operator $\hat\omega$ has particular relevance for metrological considerations. Indeed, it is the generator of time evolution. As such, the quantum Cramér-Rao bound \cite{braunstein_statistical_1994,braunstein_generalized_1996} implies that the spectral variance $\Delta\hat \omega$ is a fundamental measure of the absolute precision of time delay measurement using the state $\ket{\psi}$ as a probe.\\

The case of $n$ independent time-frequency single photon states is described by the space $\mathcal S_n=\mathcal S^{\otimes n}$. We denote the corresponding collective operator
\begin{equation}
    \hat\Omega=\hat\omega_1+\cdots+\hat\omega_n.
\end{equation}

As previously mentioned, indices are added on the operators to specify on which system they are acting. It is worthwhile to note that although we use the less known example of time-frequency variables for the study of continuous variable systems, the same argument also holds for the usual quadrature variables $(\hat p,\hat q)$. Indeed, as shown in \cite{fabre_time_2022}, both formalisms can be mapped one to the other through $\hat \omega\leftrightarrow \hat p$ and $\hat t\leftrightarrow \hat q$. 

\subsection{Spectral space}\label{subsec: spectral space}
 To provide graphical representations, we introduce what we call the spectral space. We first make a general definition and then discuss how this relates to the two examples of finite dimension and time-frequency states. Let's denote by $\Lambda\subset \R$ the spectrum of $\hat H$. As we have assumed that it is non-degenerate, to $\lambda\in\Lambda$ we can associate the corresponding eigenstate $\ket{\psi_\lambda}$. All the states of the form $\ket{\psi_{\lambda_1}}\otimes\cdots\otimes\ket{\psi_{\lambda_n}}$ for $\lambda_i\in\Lambda$ provide a basis of $\mathcal H^{\otimes n}$. As such any state $\ket{\psi}$ in $\mathcal H^{\otimes n}$ can be expanded as
 \begin{equation}
     \ket{\psi}=\int\limits_{\lambda_1,\dots,\lambda_n\in\Lambda} S(\lambda_1,\dots,\lambda_n)\ket{\psi_{\lambda_1}}\otimes\cdots\otimes\ket{\psi_{\lambda_n}} \dd^n \va{\lambda}.
 \end{equation}
To simplify the notation, we write for $\va{\lambda}=(\lambda_1,\dots,\lambda_n)\in\Lambda^n$, $\ket{\psi_{\va{\lambda}}}=\ket{\psi_{\lambda_1}}\otimes\cdots\otimes\ket{\psi_{\lambda_n}}$. Notice that in the case of a finite-dimensional $\mathcal H$, the integral corresponds to a discrete sum. Here $S: \Lambda^n\to \C$ is a square-normalized function. We call the set $\Lambda^n$ the spectral space associated with $\hat H_\text{coll}$ and the function $S$ the spectral amplitude. For small values of $n$, the function $S$ can be plotted. As we will see below, the exact values taken by $S$ do not matter for our studies. Only the knowledge of the support of $S$ ({\it i.e.}, the values of $\lambda_1,\dots,\lambda_n$ for which $S$ is non-zero) is relevant. Below, we will make such plots for the value of $n=2$.\\

It is important to notice that the spectral intensity $\abs*{S(\va{\lambda})}^2$ can be seen as a probability distribution over $\Lambda^n$. In this context, expectation values of quantum observables that are diagonal in the basis $\left\{\ket{\psi_{\va{\lambda}}}\right\}$ can be seen as the expectation value of classical random variables defined over the spectral space $\Lambda^n$. More explicitly, to each operator $\hat H_i$ we associate the random variable $H_i:\Lambda^n\to \R$ defined by $H_i(\va{\lambda})=\lambda_i$. As for collective operators defined by linear combinations, the same operations can be used to define the corresponding variable. For example, $H_\text{coll}=\sum_{i=1}^n H_i$ is the variable associated to the quantum operator $\hat H_\text{coll}$. An immediate computation shows that for a state $\ket{\psi}$ with spectral amplitude $S$
\begin{equation}
    \bra{\psi}\hat A\ket{\psi}= \expval{A},
\end{equation}
where the left-hand side is the quantum expectation value of an operator $\hat A$ diagonal in the basis $\left\{\ket{\psi_{\va{\lambda}}}\right\}$ and the right-hand-side is the expectation value of associated classical random variable $A$, expressed as $\expval{A}=\int_{\va{\lambda}\in\Lambda^n} A(\va{\lambda})\abs*{S(\va{\lambda})}^2\dd^n \va{\lambda}$. An important example of this is given by the following. For any normalized vector $\va{v}=(v_1,\dots,v_n)$ we can associate an operator $\hat V=\sum_{i=1}^n v_i \hat H_i$ and the corresponding variable $V=\sum_{i=1}^n v_i H_i$. As such the quantum expectation value equals $\expval*{\hat V}=\expval{V}$.\\

We denote by $\Delta_{\ket{\psi}}(\hat A)= \bra{\psi}\hat A^2\ket{\psi}-\bra{\psi}\hat A\ket{\psi}^2$ the quantum variance of $\hat A$ taken on the pure state $\ket{\psi}$. As they are built from expectation values, the same goes for variances: we have $\Delta \hat V=\Delta V$. As the right-hand side is simply the variance of the distribution $\abs*{S}^2$ in the direction given by $\va{v}$, this gives a simple intuition for the term $\Delta \hat V $: it is a measure of the ``thickness" of the spectral amplitude in the given direction. This will be particularly relevant for the operator $\hat H_i$ (local direction) and $\hat H_\text{coll}$ (collective or diagonal directions). \\

In the case of finite-dimensional $\mathcal H$, the spectral space $\Lambda^n$ is a discrete subspace of $\R^n$. It corresponds to an irregular finite rectangular grid. In this case, the graphical representation of the support of the spectral amplitude $S$ of a state 
$\ket{\psi}$ is a finite collection of points (see Fig.~\ref{fig: example spectral space} left for $n=2$). Opposite to this situation is the case of time-frequency states. Indeed, the spectrum of $\hat \omega$ is $\Lambda=\R$ (the eigenstate associated with the eigenvalue $\lambda$ is the state $\ket{\lambda}$ of a photon at frequency $\lambda$). In this case, the function $S$ is defined over $\R^n$. This means that the representation of the support of a time-frequency state can have any continuous shape (see Fig.~\ref{fig: example spectral space} Right for $n=2$).
\begin{figure*}
    \centering
    \scalebox{1.1}{\begin{tikzpicture}
	\begin{pgfonlayer}{nodelayer}
		\node [style=none] (0) at (-1, 0) {};
		\node [style=none] (1) at (0, -1) {};
		\node [style=none] (2) at (0, 9.5) {};
		\node [style=none] (3) at (9.5, 0) {};
		\node [style=none] (4) at (1.5, -0.25) {};
		\node [style=none] (5) at (1.5, 0) {};
		\node [style=none] (6) at (3.5, -0.25) {};
		\node [style=none] (7) at (3.5, 0) {};
		\node [style=none] (8) at (4.5, -0.25) {};
		\node [style=none] (9) at (4.5, 0) {};
		\node [style=none] (10) at (8, -0.25) {};
		\node [style=none] (11) at (8, 0) {};
		\node [style=none] (12) at (1.5, -0.25) {};
		\node [style=none] (13) at (0, 1.5) {};
		\node [style=none] (14) at (-0.25, 1.5) {};
		\node [style=none] (15) at (0, 3.5) {};
		\node [style=none] (16) at (-0.25, 3.5) {};
		\node [style=none] (17) at (0, 4.5) {};
		\node [style=none] (18) at (-0.25, 4.5) {};
		\node [style=none] (19) at (0, 8) {};
		\node [style=none] (20) at (-0.25, 8) {};
		\node [style=none] (25) at (1.25, 4.5) {};
		\node [style=none] (26) at (1.75, 4.5) {};
		\node [style=none] (27) at (1.5, 4.75) {};
		\node [style=none] (28) at (1.5, 4.25) {};
		\node [style=none] (29) at (1.25, 3.5) {};
		\node [style=none] (30) at (1.75, 3.5) {};
		\node [style=none] (31) at (1.5, 3.75) {};
		\node [style=none] (32) at (1.5, 3.25) {};
		\node [style=none] (49) at (3.25, 1.5) {};
		\node [style=none] (50) at (3.75, 1.5) {};
		\node [style=none] (51) at (3.5, 1.75) {};
		\node [style=none] (52) at (3.5, 1.25) {};
		\node [style=none] (53) at (4.25, 8) {};
		\node [style=none] (54) at (4.75, 8) {};
		\node [style=none] (55) at (4.5, 8.25) {};
		\node [style=none] (56) at (4.5, 7.75) {};
		\node [style=none] (57) at (4.25, 4.5) {};
		\node [style=none] (58) at (4.75, 4.5) {};
		\node [style=none] (59) at (4.5, 4.75) {};
		\node [style=none] (60) at (4.5, 4.25) {};
		\node [style=none] (61) at (4.25, 3.5) {};
		\node [style=none] (62) at (4.75, 3.5) {};
		\node [style=none] (63) at (4.5, 3.75) {};
		\node [style=none] (64) at (4.5, 3.25) {};
		\node [style=none] (65) at (4.25, 1.5) {};
		\node [style=none] (66) at (4.75, 1.5) {};
		\node [style=none] (67) at (4.5, 1.75) {};
		\node [style=none] (68) at (4.5, 1.25) {};
		\node [style=none] (73) at (7.75, 4.5) {};
		\node [style=none] (74) at (8.25, 4.5) {};
		\node [style=none] (75) at (8, 4.75) {};
		\node [style=none] (76) at (8, 4.25) {};
		\node [style=none] (77) at (7.75, 3.5) {};
		\node [style=none] (78) at (8.25, 3.5) {};
		\node [style=none] (79) at (8, 3.75) {};
		\node [style=none] (80) at (8, 3.25) {};
		\node [style=none] (81) at (7.75, 1.5) {};
		\node [style=none] (82) at (8.25, 1.5) {};
		\node [style=none] (83) at (8, 1.75) {};
		\node [style=none] (84) at (8, 1.25) {};
		\node [style=none] (85) at (0, 10.25) {$\lambda_2$};
		\node [style=none] (86) at (10.25, 0) {$\lambda_1$};
		\node [style=red dot] (87) at (1.5, 1.5) {};
		\node [style=red dot] (88) at (3.5, 3.5) {};
		\node [style=red dot] (89) at (3.5, 4.5) {};
		\node [style=red dot] (90) at (8, 8) {};
		\node [style=none] (102) at (11.25, 0) {};
		\node [style=none] (103) at (12.25, -1) {};
		\node [style=none] (104) at (12.25, 9.5) {};
		\node [style=none] (105) at (21.75, 0) {};
		\node [style=none] (171) at (12.25, 10.25) {$\omega_2$};
		\node [style=none] (172) at (22.5, 0) {$\omega_1$};
		\node [style=none] (173) at (14, 0.5) {};
		\node [style=none] (174) at (16, 4) {};
		\node [style=none] (175) at (14, 3) {};
		\node [style=none] (176) at (17.75, 2) {};
		\node [style=red dot] (185) at (1.5, 8) {};
		\node [style=red dot] (186) at (3.5, 8) {};
		\node [style=none] (187) at (15.5, 6.5) {};
		\node [style=none] (188) at (18.5, 6.75) {};
		\node [style=none] (189) at (19.25, 3.75) {};
		\node [style=none] (190) at (21, 6.75) {};
		\node [style=none] (191) at (15.5, 10.25) {};
		\node [style=none] (192) at (16.25, 9.5) {};
		\node [style=none] (193) at (16.75, 10) {};
		\node [style=none] (194) at (16.5, 10.75) {};
	\end{pgfonlayer}
	\begin{pgfonlayer}{edgelayer}
		\draw [style=Thick arrow] (0.center) to (3.center);
		\draw [style=Thick arrow] (1.center) to (2.center);
		\draw (5.center) to (12.center);
		\draw (7.center) to (6.center);
		\draw (9.center) to (8.center);
		\draw (11.center) to (10.center);
		\draw (20.center) to (19.center);
		\draw (18.center) to (17.center);
		\draw (16.center) to (15.center);
		\draw (14.center) to (13.center);
		\draw [style=Gray line] (27.center) to (28.center);
		\draw [style=Gray line] (25.center) to (26.center);
		\draw [style=Gray line] (31.center) to (32.center);
		\draw [style=Gray line] (29.center) to (30.center);
		\draw [style=Gray line] (51.center) to (52.center);
		\draw [style=Gray line] (49.center) to (50.center);
		\draw [style=Gray line] (55.center) to (56.center);
		\draw [style=Gray line] (53.center) to (54.center);
		\draw [style=Gray line] (59.center) to (60.center);
		\draw [style=Gray line] (57.center) to (58.center);
		\draw [style=Gray line] (63.center) to (64.center);
		\draw [style=Gray line] (61.center) to (62.center);
		\draw [style=Gray line] (67.center) to (68.center);
		\draw [style=Gray line] (65.center) to (66.center);
		\draw [style=Gray line] (75.center) to (76.center);
		\draw [style=Gray line] (73.center) to (74.center);
		\draw [style=Gray line] (79.center) to (80.center);
		\draw [style=Gray line] (77.center) to (78.center);
		\draw [style=Gray line] (83.center) to (84.center);
		\draw [style=Gray line] (81.center) to (82.center);
		\draw [style=Thick arrow] (102.center) to (105.center);
		\draw [style=Thick arrow] (103.center) to (104.center);
		\draw [style=Fill red] (174.center)
			 to [in=-150, out=90, looseness=0.75] (187.center)
			 to [in=135, out=15, looseness=0.75] (188.center)
			 to [in=0, out=-45, looseness=0.75] (176.center)
			 to [in=-45, out=165, looseness=0.25] (173.center)
			 to [in=-135, out=135, looseness=0.25] (175.center)
			 to [in=-90, out=45, looseness=0.50] cycle;
		\draw [style=Fill red] (191.center)
			 to [in=135, out=0, looseness=1.25] (194.center)
			 to [in=150, out=-30, looseness=1.25] (193.center)
			 to [in=0, out=-15, looseness=1.25] (192.center)
			 to [in=-180, out=180, looseness=1.75] cycle;
		\draw [style=Fill red] (190.center)
			 to [bend left=90, looseness=0.50] (189.center)
			 to [bend left=90, looseness=0.25] cycle;
	\end{pgfonlayer}
\end{tikzpicture}}
    \caption{Examples of the representation of the support of the spectrum in the bipartite $n=2$ case. On the left: the case of finite dimension. On the right: the case of time-frequency variables. On the left panel, crosses represent the spectral space (all the possible couples $(\lambda_1,\lambda_2)$ of eigenvalues of $\hat H$), and the red filled circles represent the support of a particular state.}
    \label{fig: example spectral space}
\end{figure*}
 
\section{Entanglement quantifier}\label{sec: entanglement measure}
With the definition of the setting in place, we can now go to the crux of the paper.
\subsection{Inequality and definition}\label{subsec: inequality and definition}
For any pure state $\ket{\psi}\in\mathcal H^{\otimes n}$ we can verify that
\begin{equation}\label{eq: general bound collective global}
    \Delta(\hat H_\text{coll})\leq n^2 \max_i  \Delta(\hat H_i),
\end{equation}
For clarity reasons, when no confusion is possible, the dependence on the quantum state is kept implicit. This inequality establishes a link between the variance of a quantum state in various directions. Notably, the collective variance is bounded by the local ones. As this expression is central to this study, its meaning and implications will be discussed in more detail in later parts of this paper. We first proceed with its proof. Although a simple proof can be made using the Cauchy-Schwarz inequality\footnote{$\Delta(\hat H_\text{coll})=\sum_{i,j=1}^n\operatorname{Cov}(\hat H_i,\hat H_j)\leq \sum_{i,j=1}^n\sqrt{\Delta(\hat H_i)\Delta(\hat H_j)}\leq n^2\max_i\Delta(\hat H_i)$} on the bilinear form of quantum covariance, we provide an alternate proof as it uses a formalism that will be used later.\\

Let's introduce an orthogonal matrix $A=(\alpha_{ij})\in\mathcal O_n(\R)$ for which we have $\alpha_{1j}=\frac{1}{\sqrt{n}}$ for all index $j$. As $A$ is orthogonal, we have the following relations coming from $AA^T=A^TA=I_n$ (the $n\times n$ identity matrix)
\begin{align}
    \sum_{i=1}^n \alpha_{ij}\alpha_{ik}=\delta_{jk}, && \sum_{i=1}^n\alpha_{ji}\alpha_{ki}=\delta_{jk}.
\end{align}
With the matrix $A$, we introduce new collective operators
\begin{equation}\label{eq: definition P operators}
    \hat P_i=\sum_{j=1}^n \alpha_{ij}\hat H_j.
\end{equation}
Given the condition on the first line of $A$, we have that
\begin{equation}
    \hat P_1=\sum_{j=1}^n \alpha_{1j} \hat H_j=\frac{1}{\sqrt{n}}\hat H_\text{coll}.
\end{equation}
Finally, with the inversion property of $A$, we can express the $\hat H_j$'s operators in terms of the $\hat P_i$'s
\begin{equation}
    \sum_{i=1}^n \alpha_{ij}\hat P_i=\sum_{i,k=1}^n\alpha_{ij}\alpha_{ik}\hat H_k=\sum_{k=1}^n \delta_{jk}\hat H_k=\hat H_j.
\end{equation}
With all these definitions/expressions in place, we can compute
\begin{subequations}
    \begin{align}
        \max_i\Delta(\hat H_i)&\geq \frac{1}{n}\sum_{i=1}^n \Delta(\hat H_i)\\
        &=\frac{1}{n}\sum_{i,j,k=1}^n\alpha_{ji}\alpha_{ki}\operatorname{Cov}(\hat P_j,\hat P_k)\\
        &=\frac{1}{n}\sum_{j,k=1}^n\delta_{jk}\operatorname{Cov}(\hat P_j,\hat P_k)\\
        &=\frac{1}{n}\sum_{j=1}^n\Delta(\hat P_j)\\
        &\geq \frac{1}{n}\Delta(\hat P_1)=\frac{1}{n^2}\Delta(\hat H_\text{coll}),\label{eq: last step proof inequ}
    \end{align}
\end{subequations}
where, $ \operatorname{Cov}_{\ket{\psi}}(\hat A,\hat B)= \bra{\psi}\hat A\hat B\ket{\psi}-\bra{\psi}\hat A\ket{\psi}\bra{\psi}\hat B\ket{\psi}$, is the quantum covariance for commuting operators. The first inequality is obtained by saying that the maximum is larger than the mean, and the last one is because the variance of $\hat P_j$ ($j>1$) is non-negative. By multiplying by $n^2$ we obtain the desired inequality. Applied to the time-frequency setting, we obtain a stronger version of the inequality from Ref.~\cite{descamps_quantum_2023}. In the case of finite-dimensional Hilbert spaces, one can simplify the inequality~(\ref{eq: general bound collective global}). Indeed, the right-hand side can be bounded from above by a state-independent term. By denoting $\lambda_{\max}$ (resp. $\lambda_{\min}$) the largest (resp. smallest) eigenvalues of $\hat H$, one can verify that for any state $\ket{\psi}\in\mathcal H^{\otimes n}$ \cite{pezze_quantum_2014} (refer to appendix~\ref{appendix popoviciu's} for a proof provided for completeness), 
\begin{equation}
    \Delta_{\ket{\psi}}(\hat H_i)\leq \frac{(\lambda_{\max}-\lambda_{\min})^2}{4},
\end{equation}
with this, we recover equality already formulated within this general context in Ref.~\cite{pezze_quantum_2014}

\begin{equation}\label{eq: general bound in finite dim}
    \Delta(\hat H_\text{coll})\leq \frac{n^2(\lambda_{\max}-\lambda_{\min})^2}{4}.
\end{equation}
Applying this for the case of spin observables, where $\lambda_{\max}=-\lambda_{\min}=1$ recovers the usual sum of spin inequality $\Delta(\hat Z_1+\cdots +\hat Z_n)\leq n^2$.\\

Taking the quotient of the right and left sides of the inequality~(\ref{eq: general bound collective global}), it is natural to define the following quantity
\begin{equation}\label{eq: def quantifier}
    F_{\ket{\psi}}=\frac{\Delta(\hat H_\text{coll})}{\max_i  \Delta(\hat H_i)},
\end{equation}
provided that the denominator does not vanish. If the denominator vanishes, Eq.~(\ref{eq: general bound collective global})  implies that the numerator must also vanish, resulting in an indeterminate form. By convention, we define $F=0$ for such states. Using the additivity of the quantum variance for pure product states \cite{pezze_quantum_2014}, we verify that $F$ is sub-additive on product pure states. If $\ket{\psi}\in\mathcal H^{\otimes l}$ and $\ket{\phi}\in\mathcal H^{\otimes m}$, we have
\begin{subequations}\label{eq: subadditivity}
    \begin{align}
        F_{\ket{\psi}\otimes\ket{\phi}}&=\frac{\Delta_{\ket{\psi}\otimes\ket{\phi}} \hat H_\text{coll}}{\max_i\Delta_{\ket{\psi}\otimes\ket{\phi}} \hat H_i}\\
        &=\frac{\Delta_{\ket{\psi}}(\hat H_1+\cdots+\hat H_l)+\Delta_{\ket{\phi}}(\hat H_1+\cdots+\hat H_m)}{\max\{\max_i \Delta_{\ket{\psi}}\hat H_i,\max_j \Delta_{\ket{\phi}}\hat H_j\}}\\
        &\leq \frac{\Delta_{\ket{\psi}}(\hat H_1+\cdots+\hat H_l)}{\max_i \Delta_{\ket{\psi}}\hat H_i}+\frac{\Delta_{\ket{\phi}}(\hat H_1+\cdots+\hat H_m)}{\max_j \Delta_{\ket{\phi}}\hat H_j}\\
        &=F_{\ket{\psi}}+F_{\ket{\phi}}
    \end{align}    
\end{subequations}

Recursively applying this to a product state $\ket{\psi}=\ket{\psi_1}\otimes\cdots\otimes\ket{\psi_n}$ and noting that for each single-system state $\ket{\psi_i}\in\mathcal H$, $F_{\ket{\psi_i}}=1$, it follows that for any product state, $F_{\ket{\psi}}\leq n$. As such, obtaining $F_{\ket{\psi}}> n$ is a witness of entanglement. Entanglement witnesses are typically defined as functions $\mathcal{F}(\hat{\rho})$ that satisfy $\mathcal{F}(\hat{\rho}) \geq 0$ for all separable states and $\mathcal{F}(\hat{\rho}) < 0$ for some entangled states. In our case, $\mathcal{F}(\hat{\rho}) = n - F$ fits to this characterization. It is worth noting that the simplest form of an entanglement witness is linear, taking the form $\mathcal{F}(\hat{\rho}) = \Tr(\hat{W}\hat{\rho})$, where $\hat{W}$ is a suitably chosen observable. However, our witness does not belong to the category of linear witnesses. Non-linear witnesses form a richer class of witnesses that has already been explored, for example, in \cite{hyllus_optimal_2006,guhne_nonlinear_2006}. Entanglement witnesses can be merely mathematical constructs or be related to some physical property of the system. In our work, it is indeed possible to assign a physical interpretation to $F$. Specifically, in Sec.~\ref{subsec: motivation}, we explore its connection to metrology, and throughout the paper, we analyze its graphical interpretation. Intuitively, it quantifies how much the state is entangled in the collective variable $H_\text{coll}$. In the following, we aim to make this intuition more robust.

\subsection{Motivation}\label{subsec: motivation}
Now that we have covered all the foundational aspects of our work, we turn to the discussion of the motivations behind our definitions and the physical implications of the bound in equation (\ref{eq: general bound collective global}), as well as the states that saturate them. First, it is important to recognize that our work, like many previous studies, explores the connections between entanglement and metrology. In its simplest form, where $n$ systems are used to estimate a parameter $\theta$ encoded in each system via the evolution $e^{i\theta\hat H_i}$, the ultimate precision limit for pure states is directly related to the quantum variance of the operator $\hat H_\text{coll}=\hat H_1+\cdots+\hat H_n$, or the quantum Fisher information when dealing with mixed states. For independent probes, the precision scales linearly with $n$, a behavior known as {\it shot noise}. This is analogous to our statement that $F\leq n$ for product states. The key insight of quantum metrology is that entanglement between different parties can enable higher precision. In such cases, the precision can scale up to $n^2$, a phenomenon referred to as {\it Heisenberg scaling}. Therefore, equation (\ref{eq: general bound collective global}) can be interpreted as a generalized version of the Heisenberg bound, valid for an extensive class of systems, which asserts that optimal precision always scales quadratically with the number of probes.\\

The use of the quantum Fisher information for detecting and measuring entanglement arises from the following observation: if entanglement is necessary for achieving a high value of the quantum Fisher information, then, conversely, a large value of the quantum Fisher information serves as an entanglement witness. This idea has been explored in works such as \cite{hyllus_fisher_2012,chen_wigner-yanase_2005, pezze_quantum_2014}. In this paper, we further develop this concept. We recognize that, given an initial state $\ket{\psi}$, one can generally increase the global variance $\Delta_{\ket{\psi}}\hat H_\text{coll}$ by modifying $\ket{\psi}$. The first approach is to ``stretch" the state by increasing the local variance $\Delta\hat H_i$ without altering the amount of entanglement. The second approach is to increase entanglement without changing the local variances $\Delta\hat H_i$. Both of these possibilities are illustrated in Fig.~\ref{fig: increasing_precision}. Although such enlargements are not always feasible, depending on the initial state and the Hilbert space $\mathcal H$, they demonstrate that the global variance $\Delta\hat H_\text{coll}$ is not exclusively a quantifier of entanglement, as it can be increased without introducing additional entanglement. This is why we define $F$ as the ratio of the global variance $\Delta\hat H_\text{coll}$ to the maximum of the local variances $\max_i \Delta\hat H_i$, ensuring that the stretching of the state (which does not impact entanglement) does not affect the value of $F$. The usefulness of comparing local variances and the quantum Fisher information of non-local operators was already highlighted in \cite{gessner_efficient_2016}, where a general method to construct entanglement criteria based on local observables was introduced, and further developed in \cite{lopetegui_detection_2024} to design experimentally accessible witnesses capable of detecting strong forms of entanglement in non-Gaussian continuous-variable states. 

\begin{figure}
    \centering
    \scalebox{0.7}{\begin{tikzpicture}
	\begin{pgfonlayer}{nodelayer}
		\node [style=none] (83) at (-15, -3) {};
		\node [style=none] (84) at (-14, -4) {};
		\node [style=none] (85) at (-14, 3.5) {};
		\node [style=none] (86) at (-8, -3) {};
		\node [style=none] (87) at (-14, 4.25) {$\lambda_2$};
		\node [style=none] (88) at (-6.75, -2.5) {$\lambda_1$};
		\node [style=none] (106) at (-1, -7.5) {};
		\node [style=none] (107) at (0, -8.5) {};
		\node [style=none] (108) at (0, -2.5) {};
		\node [style=none] (109) at (8.25, -7.5) {};
		\node [style=none] (112) at (-1, 0) {};
		\node [style=none] (113) at (0, -1) {};
		\node [style=none] (114) at (0, 7.75) {};
		\node [style=none] (115) at (8.25, 0) {};
		\node [style=none] (126) at (-13, -2.5) {};
		\node [style=none] (127) at (-10, -2.5) {};
		\node [style=none] (128) at (-9, -2.5) {};
		\node [style=red dot] (129) at (-13, 2.5) {};
		\node [style=red dot] (130) at (-9, 2.5) {};
		\node [style=red dot] (131) at (-13, -1.5) {};
		\node [style=red dot] (132) at (-9, -1.5) {};
		\node [style=none] (133) at (-12, -0.5) {};
		\node [style=none] (134) at (-11, 0.5) {};
		\node [style=none] (135) at (-10, 1.5) {};
		\node [style=none] (136) at (-11, -2) {$\Delta\hat H_1$};
		\node [style=none] (137) at (-12, 1) {$\Delta\hat  H_\text{coll}$};
		\node [style=none] (138) at (1, 0.25) {};
		\node [style=none] (139) at (4, 0.25) {};
		\node [style=none] (140) at (7, 0.25) {};
		\node [style=red dot] (141) at (1, 7.25) {};
		\node [style=red dot] (142) at (7, 7.25) {};
		\node [style=red dot] (143) at (1, 1.25) {};
		\node [style=red dot] (144) at (7, 1.25) {};
		\node [style=none] (148) at (3, 0.75) {$\Delta\hat H_1$};
		\node [style=none] (150) at (2, -7) {};
		\node [style=none] (151) at (5, -7) {};
		\node [style=none] (152) at (6, -7) {};
		\node [style=red dot] (154) at (6, -2.5) {};
		\node [style=red dot] (155) at (2, -6.5) {};
		\node [style=none] (157) at (2.25, -6.25) {};
		\node [style=none] (158) at (4, -4.5) {};
		\node [style=none] (159) at (5.75, -2.75) {};
		\node [style=none] (160) at (4, -6.5) {$\Delta\hat H_1$};
		\node [style=none] (161) at (3, -4) {$\Delta\hat H_\text{coll}$};
		\node [style=none] (162) at (2.25, 2.5) {};
		\node [style=none] (163) at (4, 4.25) {};
		\node [style=none] (164) at (5.75, 6) {};
		\node [style=none] (165) at (3, 4.75) {$\Delta\hat  H_\text{coll}$};
		\node [style=none] (166) at (0, 8.5) {$\lambda_2$};
		\node [style=none] (167) at (0, -1.75) {$\lambda_2$};
		\node [style=none] (168) at (9, -7.5) {$\lambda_1$};
		\node [style=none] (169) at (9, 0) {$\lambda_1$};
		\node [style=none] (170) at (-6, 1) {};
		\node [style=none] (171) at (-1, 4) {};
		\node [style=none] (172) at (-6, -1) {};
		\node [style=none] (173) at (-1, -4) {};
		\node [style=none] (174) at (-4, 2.75) {\rotatebox{30.9638}{\scalebox{0.8}{Stretching the state}}};
		\node [style=none] (175) at (-3.25, -2) {\rotatebox{-30.9638}{\scalebox{0.8}{Increasing entanglement}}};
	\end{pgfonlayer}
	\begin{pgfonlayer}{edgelayer}
		\draw [style=Thick arrow] (83.center) to (86.center);
		\draw [style=Thick arrow] (84.center) to (85.center);
		\draw [style=Thick arrow] (106.center) to (109.center);
		\draw [style=Thick arrow] (107.center) to (108.center);
		\draw [style=Thick arrow] (112.center) to (115.center);
		\draw [style=Thick arrow] (113.center) to (114.center);
		\draw [style=Arrow] (127.center) to (128.center);
		\draw [style=Arrow] (127.center) to (126.center);
		\draw [style=Arrow] (134.center) to (135.center);
		\draw [style=Arrow] (134.center) to (133.center);
		\draw [style=Arrow] (139.center) to (140.center);
		\draw [style=Arrow] (139.center) to (138.center);
		\draw [style=Arrow] (151.center) to (152.center);
		\draw [style=Arrow] (151.center) to (150.center);
		\draw [style=Arrow] (158.center) to (159.center);
		\draw [style=Arrow] (158.center) to (157.center);
		\draw [style=Arrow] (163.center) to (164.center);
		\draw [style=Arrow] (163.center) to (162.center);
		\draw [style=Thick arrow] (170.center) to (171.center);
		\draw [style=Thick arrow] (172.center) to (173.center);
	\end{pgfonlayer}
\end{tikzpicture}}
    \caption{Two ways to increase the collective variance $\Delta\hat H_\text{coll}$ of a pure state. The left plot shows the spectral representation of an initial product state $\ket{\psi}$. The upper-right plot depicts the same state after it has been stretched. Although the state remains a product state, it now exhibits a larger $\Delta\hat H_\text{coll}$, at the cost of similarly increasing the local variance $\Delta\hat H_i$. The lower-right plot illustrates a state of similar size to the original, but with certain points removed, resulting in an entangled state with an unchanged local variance $\Delta\hat H_i$. In both cases, the collective variance $\Delta\hat H_\text{coll}$ is increased to the same extent.}
    \label{fig: increasing_precision}
\end{figure}

\subsection{Remark on the definition of collective operators}\label{subsec: remark def coll op}
Previously, we defined the collective operators $\hat H_\text{coll}$ and $\hat P_i$. Here, we briefly discuss the choices made in these definitions. Specifically, we defined $\hat H_\text{coll}=\hat H_1+\cdots+\hat H_n$ as the sum of all local operators $\hat H_i$, with the prefactors set to 1. More generally, these prefactors could be chosen as either 1 or $-1$.
\begin{equation}
    \hat H_\text{coll}=\sum_{i=1}^n c_i\hat H_i,
\end{equation}
where $c_i=\pm1$. All results presented so far, as well as those discussed later in this paper, can be readily adapted to this more general definition of the collective operator $\hat H_\text{coll}$.  However, for clarity and simplicity of exposition, we will restrict our focus to the case $c_i=1$ outside of this section.\\

The specific definition of $\hat H_\text{coll}$ depends on the system under study. For instance, the original definition ($c_i=1,\forall i$) is particularly relevant in metrological contexts, as demonstrated in \cite{pezze_quantum_2014}. The findings in \cite{pezze_quantum_2014}, together with inequality (\ref{eq: general bound collective global}), provide valuable insights into the achievable estimation precision and the characteristics of optimal states. Further discussion regarding the optimal states is presented in Sec.~\ref{subsec: equality case}.\\

The definition of $\hat{H}_\text{coll}$ involving $-1$ prefactors can also be relevant in certain contexts. One notable example is provided by time-frequency states for the case of two single photons ($n=2$). Such states can be generated through the SPDC process \cite{boucher_toolbox_2015} and, under specific conditions, exhibit a spectrum with strong anti-correlation in frequency. These states yield a high value of $F$ when considering $\hat{H}_\text{coll} = \hat{\omega}_- = \hat{\omega}_1 - \hat{\omega}_2$, while the corresponding value of $F$ for $\hat{\omega}_+ = \hat{\omega}_1 + \hat{\omega}_2$ remains close to zero. Furthermore, as demonstrated in \cite{descamps_time-frequency_2023}, the Hong-Ou-Mandel interferometer is particularly well-suited for measuring two-photon states, specifically providing information related to the operator $\hat{\omega}_-$ rather than $\hat{\omega}_+$. In this context, analyzing $F$ as defined for the collective operator $\hat{\omega}_-$ offers significant insight into the extent to which two-photon states are optimized for Hong-Ou-Mandel measurements.\\

In the proof of Eq.~(\ref{eq: general bound collective global}), we introduce a set of operators $\hat P_i$, defined in Eq.~(\ref{eq: definition P operators}), constructed using collective variables that are orthogonal to the one defining $\hat H_\text{coll}$. It is important to note that the $\hat P_i$ operators are themselves collective, as they act simultaneously on multiple systems. However, two key distinctions exist between the definitions of $\hat H_\text{coll}$ and $\hat P_i$. Firstly, the normalization of $\hat H_\text{coll}$ ensures that the sum of the squared coefficients satisfies $\sum_{i=1}^n c_i^2=n$. In contrast, the coefficients of $\hat P_i$ are normalized such that $\sum_{j=1}^n \alpha_{ij}^2=1$. To address this discrepancy, we have defined a normalized version of $\hat H_\text{coll}$, denoted as $\hat P_1=\hat H_\text{coll}/\sqrt{n}$. Additionally, in the definition of $\hat H_\text{coll}$, the coefficients $c_i$ are restricted to $\pm1$, while $\hat P_i$ permits arbitrary coefficients $\alpha_{i,j}$. The constraint $c_i=\pm1$ in $\hat H_\text{coll}$ serves two purposes. First, it ensures symmetry among all modes ($\abs{c_i}=cst$), and as shown in the appendix \ref{appendix: cauchy-schwarz general} it is the only case where states saturating the inequality in (\ref{eq: general bound collective global}) can exist.

\subsection{Equality case}\label{subsec: equality case}
In the previous section, we have argued that $F_{\ket{\psi}}$ is a quantifier for a specific entanglement related to a metrological context. A first step to make this intuition more robust is to understand for which states $\ket{\psi}$, $F_{\ket{\psi}}$ is maximal, {\it i.e.}, the states that saturate the inequality~(\ref{eq: general bound collective global}). We first discuss the general conditions such states must satisfy, then discuss the geometrical consequences in terms of the spectral space before analyzing how this works in the two examples. 
\subsubsection{General argument}\label{subsubsec: general argument}
Looking back at the previous proof, it is easy to see under which conditions the inequality Eq.~(\ref{eq: general bound collective global}) is saturated. Indeed, this bound is obtained by a chain of equalities and inequalities. As only two inequalities were used, it follows that Eq.~(\ref{eq: general bound collective global}) is saturated if and only if
\begin{equation}
    \left\{\begin{aligned}
        &\forall i,j,\;\Delta(\hat H_i)=\Delta(\hat H_j),  \\
          &\forall j\geq 2,\;\Delta(\hat P_j)=0.
    \end{aligned}\right.
\end{equation}
However, it is easy to see that the second condition automatically implies the first one. Indeed, the Cauchy-Schwarz inequality for the covariance implies that 
\begin{equation}
    \abs{\operatorname{Cov}(\hat P_i,\hat P_j)}^2\leq \Delta(\hat P_i)\Delta(\hat P_j)=0,
\end{equation}
if $i\geq 2$ or $j\geq 2$. Thus when expanding $\Delta(\hat H_i)$, in terms of the covariance of the $\hat P_j$'s operators, only one term is non-zero, and we recover $\Delta(\hat H_i)=\frac{1}{n^2}\Delta(\hat H_\text{coll})$.\\

First, a trivial class of states that saturate the bound are zero variance states, for which both the left- and right-hand sides of~(\ref{eq: general bound collective global}) are zero. They are the product of eigenvectors of $\hat H$:
\begin{equation}
    \ket{\psi}=\ket{\psi_{\lambda_1}}\otimes\cdots\otimes \ket{\psi_{\lambda_n}}.
\end{equation}
However, such states are product states for which we have, by convention, set the quantifier $F$ defined in equation (\ref{eq: def quantifier}) to zero. We will thus avoid such states in the following. Another class of states that achieve the equality in Eq.~(\ref{eq: general bound collective global}) is given by 
\begin{equation}\label{eq: first diagonal states}
    \ket{\psi}=\int\limits_{\lambda\in\Lambda} S(\lambda)\ket{\psi_\lambda}^{\otimes n} d\lambda,
\end{equation}
for an arbitrary function $S:\Lambda\to\C$. In the case of finite-dimensional $\mathcal H$, this corresponds to the sum $\sum_{\lambda\in\Lambda} C_\lambda\ket{\psi_\lambda}^{\otimes n}$ where the $C_\lambda$ are arbitrary constants. However, depending on the specific geometry of $\Lambda$, other states may exist. 

\subsubsection{Geometrical consequence}\label{subsubsec: geometrical consequence}
We consider a state $\ket{\psi}\in\mathcal H^{\otimes n}$ and denote by $S:\Lambda^n\to\C$ its spectral amplitude. As already stated, the state $\ket{\psi}$ saturates the inequality~(\ref{eq: general bound collective global}) if and only if $\Delta \hat P_j=0$ for all $j\geq 2$.  Given the geometric intuition provided by the spectral space introduced in Section~\ref{subsec: spectral space}, it is clear that the distribution $S$ has zero thickness in the direction given by the variables $P_2,\dots,P_n$, as the variance $\Delta P_j$ (for $j\geq 2$) are zero. The remaining orthogonal direction is given by the variable $P_1$ which is directed by the vector $\va{u}=(1,\dots,1)$. As such, the support of $S$ needs to lie on a line ({\it i.e.}, a one-dimensional affine subspace of $\R^n$) directed by $\va{u}$. We can note that this line does not need to pass through the origin. Depending on the spectrum $\Lambda$, there may be more or less such lines inside the spectral space $\Lambda^n$. Different lines will yield different states saturating Eq.~(\ref{eq: general bound collective global}). This is what we explore with our two examples in the next sections.

By observing that the spectral representation of the optimal state is elongated along the main diagonal, we gain a geometric insight into the type of entanglement quantified by $F$: it quantifies the extent to which the spectral representation is stretched in the direction of the variable $P_1$. Mathematically, this elongation corresponds to positive correlations among all the random variables $H_i$.

\subsubsection{For finite dimension}\label{subsubsec: finite dimension}
In the finite-dimensional case, we can study the saturation of the original inequality~(\ref{eq: general bound collective global}) or the one obtained after simplification~(\ref{eq: general bound in finite dim}).\\

\paragraph{For the general bound}\label{parag: finite dim equal for general bound}~\\
In the previous sections, we gave an example of a state saturating the bound~(\ref{eq: general bound collective global}) and gave a geometrical description of all such states.  As it is fundamentally linked to the geometry of the spectrum $\Lambda$, we now analyze which finite spectrum leads to such states. As $\Lambda$ is finite, we can enumerate its elements, which we assume are ordered $\Lambda=\{\lambda_1<\cdots<\lambda_d\}$. As such, the spectral space $\Lambda^n$ is an uneven rectangular grid of $d^n$ points.\\

As expressed in the previous section, the support of $S$ is on a line of direction vector $\va{u}=(1,\dots,1)$. In Fig.~\ref{fig:evenly spaced grid} we plot the $2$-dimensional spectral space associated with $\hat H$ with evenly spaced eigenvalues. In this simple situation, any diagonal with directing vector $\va{u}$ passing by at least two points will correspond to nontrivial states saturating the bound. However, in general, when the eigenvalues of $\hat H$ are not equally spaced, there are fewer possibilities, as shown in Fig.~\ref{fig:not evenly spaced grid}. To have more possibilities, than just the one corresponding to states of Eq.~(\ref{eq: first diagonal states}), at least three eigenvalues of $\hat H$ should be equally spaced ({\it i.e.} $\abs{\lambda_a-\lambda_b}=\abs{\lambda_c-\lambda_b}$).\\

\begin{figure}
    \centering
    \scalebox{0.6}{\begin{tikzpicture}
	\begin{pgfonlayer}{nodelayer}
		\node [style=green dot] (0) at (0, 0) {};
		\node [style=Blue dot] (3) at (12, 0) {};
		\node [style=red dot] (7) at (0, 8) {};
		\node [style=green dot] (11) at (4, 4) {};
		\node [style=red dot] (13) at (4, 12) {};
		\node [style=green dot] (17) at (8, 8) {};
		\node [style=red dot] (19) at (8, 16) {};
		\node [style=green dot] (23) at (12, 12) {};
		\node [style=Blue dot] (26) at (16, 4) {};
		\node [style=green dot] (28) at (16, 16) {};
		\node [style=none] (34) at (0, 17.5) {};
		\node [style=none] (35) at (0, -1) {};
		\node [style=none] (36) at (-1, 0) {};
		\node [style=none] (37) at (17.5, 0) {};
		\node [style=none] (86) at (-1, -1) {};
		\node [style=none] (87) at (17, 17) {};
		\node [style=none] (88) at (-1, 7) {};
		\node [style=none] (89) at (9, 17) {};
		\node [style=none] (90) at (11, -1) {};
		\node [style=none] (91) at (17, 5) {};
		\node [style=none] (92) at (-0.25, 16) {};
		\node [style=none] (93) at (0.25, 16) {};
		\node [style=none] (94) at (4, 15.75) {};
		\node [style=none] (95) at (4, 16.25) {};
		\node [style=none] (96) at (3.75, 16) {};
		\node [style=none] (97) at (4.25, 16) {};
		\node [style=none] (98) at (12, 15.75) {};
		\node [style=none] (99) at (12, 16.25) {};
		\node [style=none] (100) at (11.75, 16) {};
		\node [style=none] (101) at (12.25, 16) {};
		\node [style=none] (102) at (8, 11.75) {};
		\node [style=none] (103) at (8, 12.25) {};
		\node [style=none] (104) at (7.75, 12) {};
		\node [style=none] (105) at (8.25, 12) {};
		\node [style=none] (106) at (4, 7.75) {};
		\node [style=none] (107) at (4, 8.25) {};
		\node [style=none] (108) at (3.75, 8) {};
		\node [style=none] (109) at (4.25, 8) {};
		\node [style=none] (112) at (-0.25, 12) {};
		\node [style=none] (113) at (0.25, 12) {};
		\node [style=none] (114) at (16, 11.75) {};
		\node [style=none] (115) at (16, 12.25) {};
		\node [style=none] (116) at (15.75, 12) {};
		\node [style=none] (117) at (16.25, 12) {};
		\node [style=none] (118) at (16, 7.75) {};
		\node [style=none] (119) at (16, 8.25) {};
		\node [style=none] (120) at (15.75, 8) {};
		\node [style=none] (121) at (16.25, 8) {};
		\node [style=none] (122) at (12, 7.75) {};
		\node [style=none] (123) at (12, 8.25) {};
		\node [style=none] (124) at (11.75, 8) {};
		\node [style=none] (125) at (12.25, 8) {};
		\node [style=none] (126) at (12, 3.75) {};
		\node [style=none] (127) at (12, 4.25) {};
		\node [style=none] (128) at (11.75, 4) {};
		\node [style=none] (129) at (12.25, 4) {};
		\node [style=none] (130) at (8, 3.75) {};
		\node [style=none] (131) at (8, 4.25) {};
		\node [style=none] (132) at (7.75, 4) {};
		\node [style=none] (133) at (8.25, 4) {};
		\node [style=none] (134) at (-0.25, 4) {};
		\node [style=none] (135) at (0.25, 4) {};
		\node [style=none] (136) at (8, -0.25) {};
		\node [style=none] (137) at (8, 0.25) {};
		\node [style=none] (138) at (4, -0.25) {};
		\node [style=none] (139) at (4, 0.25) {};
		\node [style=none] (140) at (16, -0.25) {};
		\node [style=none] (141) at (16, 0.25) {};
		\node [style=none] (142) at (4.5, 9.5) {};
		\node [style=none] (143) at (6.5, 11.5) {};
		\node [style=none] (144) at (5, 10.75) {$\va{u}$};
	\end{pgfonlayer}
	\begin{pgfonlayer}{edgelayer}
		\draw [style=Thick arrow] (35.center) to (34.center);
		\draw [style=Thick arrow] (36.center) to (37.center);
		\draw [style=dashes] (86.center) to (87.center);
		\draw [style=dashes] (88.center) to (89.center);
		\draw [style=dashes] (90.center) to (91.center);
		\draw (92.center) to (93.center);
		\draw (95.center) to (94.center);
		\draw (96.center) to (97.center);
		\draw (99.center) to (98.center);
		\draw (100.center) to (101.center);
		\draw (103.center) to (102.center);
		\draw (104.center) to (105.center);
		\draw (107.center) to (106.center);
		\draw (108.center) to (109.center);
		\draw (112.center) to (113.center);
		\draw (115.center) to (114.center);
		\draw (116.center) to (117.center);
		\draw (119.center) to (118.center);
		\draw (120.center) to (121.center);
		\draw (123.center) to (122.center);
		\draw (124.center) to (125.center);
		\draw (127.center) to (126.center);
		\draw (128.center) to (129.center);
		\draw (131.center) to (130.center);
		\draw (132.center) to (133.center);
		\draw (134.center) to (135.center);
		\draw (137.center) to (136.center);
		\draw (139.center) to (138.center);
		\draw (141.center) to (140.center);
		\draw [style=Thick arrow] (142.center) to (143.center);
	\end{pgfonlayer}
\end{tikzpicture}}
    \caption{Spectral space of $\mathcal H^{\otimes 2}$ in the case of evenly spaced eigenvalues. In red, green, and blue are given examples of support for states saturating the inequality. Crosses represent the spectral space (all the possible couple $(\lambda_1,\lambda_2)$ of eigenvalues of $\hat H$).}
    \label{fig:evenly spaced grid}
\end{figure}

\begin{figure}
    \centering
    \scalebox{0.6}{\begin{tikzpicture}
	\begin{pgfonlayer}{nodelayer}
		\node [style=green dot] (0) at (0, 0) {};
		\node [style=green dot] (1) at (6, 0) {};
		\node [style=green dot] (2) at (9, 0) {};
		\node [style=green dot] (3) at (12, 0) {};
		\node [style=green dot] (4) at (17, 0) {};
		\node [style=green dot] (6) at (0, 6) {};
		\node [style=green dot] (7) at (0, 9) {};
		\node [style=green dot] (8) at (0, 12) {};
		\node [style=green dot] (9) at (0, 17) {};
		\node [style=green dot] (11) at (6, 6) {};
		\node [style=green dot] (12) at (6, 9) {};
		\node [style=green dot] (13) at (6, 12) {};
		\node [style=green dot] (14) at (6, 17) {};
		\node [style=green dot] (16) at (9, 6) {};
		\node [style=green dot] (17) at (9, 9) {};
		\node [style=green dot] (18) at (9, 12) {};
		\node [style=green dot] (19) at (9, 17) {};
		\node [style=green dot] (21) at (12, 6) {};
		\node [style=green dot] (22) at (12, 9) {};
		\node [style=green dot] (23) at (12, 12) {};
		\node [style=green dot] (24) at (12, 17) {};
		\node [style=green dot] (26) at (17, 6) {};
		\node [style=green dot] (27) at (17, 12) {};
		\node [style=green dot] (28) at (17, 17) {};
		\node [style=none] (34) at (0, 18.5) {};
		\node [style=none] (35) at (0, -1) {};
		\node [style=none] (36) at (-1, 0) {};
		\node [style=none] (37) at (18.5, 0) {};
		\node [style=none] (40) at (16, -1) {};
		\node [style=none] (41) at (11, -1) {};
		\node [style=none] (42) at (8, -1) {};
		\node [style=none] (43) at (5, -1) {};
		\node [style=none] (44) at (-1, -1) {};
		\node [style=none] (45) at (18, 1) {};
		\node [style=none] (47) at (18, 6) {};
		\node [style=none] (48) at (18, 15) {};
		\node [style=none] (49) at (18, 18) {};
		\node [style=none] (53) at (18, 9) {};
		\node [style=none] (54) at (18, 12) {};
		\node [style=none] (57) at (2, -1) {};
		\node [style=none] (63) at (15, 18) {};
		\node [style=none] (64) at (-1, 2) {};
		\node [style=none] (65) at (-1, 5) {};
		\node [style=none] (66) at (12, 18) {};
		\node [style=none] (67) at (-1, 8) {};
		\node [style=none] (70) at (9, 18) {};
		\node [style=none] (71) at (6, 18) {};
		\node [style=none] (75) at (1, 18) {};
		\node [style=none] (80) at (-1, 16) {};
		\node [style=none] (84) at (-1, 11) {};
		\node [style=green dot] (85) at (17, 9) {};
		\node [style=none] (86) at (-1, 10) {};
		\node [style=none] (87) at (7, 18) {};
		\node [style=none] (88) at (10, 18) {};
		\node [style=none] (89) at (-1, 7) {};
		\node [style=none] (90) at (-1, 4) {};
		\node [style=none] (91) at (13, 18) {};
		\node [style=none] (92) at (18, 10) {};
		\node [style=none] (93) at (6.75, -1) {};
		\node [style=none] (94) at (10, -1) {};
		\node [style=none] (95) at (18, 7) {};
		\node [style=none] (96) at (18, 13) {};
		\node [style=none] (97) at (4, -1) {};
		\node [style=none] (98) at (3.75, -14) {};
	\end{pgfonlayer}
	\begin{pgfonlayer}{edgelayer}
		\draw [style=Thick arrow] (35.center) to (34.center);
		\draw [style=Thick arrow] (36.center) to (37.center);
		\draw (40.center) to (45.center);
		\draw [in=225, out=45] (41.center) to (47.center);
		\draw (53.center) to (42.center);
		\draw [style=dashes red] (54.center) to (43.center);
		\draw [style=dashes red] (57.center) to (48.center);
		\draw [style=dashes red] (44.center) to (49.center);
		\draw [style=dashes red] (64.center) to (63.center);
		\draw [style=dashes red] (66.center) to (65.center);
		\draw (67.center) to (70.center);
		\draw [in=45, out=-135] (71.center) to (84.center);
		\draw (80.center) to (75.center);
		\draw (86.center) to (87.center);
		\draw (89.center) to (88.center);
		\draw (90.center) to (91.center);
		\draw (95.center) to (94.center);
		\draw (93.center) to (92.center);
		\draw (97.center) to (96.center);
	\end{pgfonlayer}
\end{tikzpicture}}
    \caption{Spectral space of $\mathcal H^{\otimes 2}$ in the case of unevenly spaced eigenvalues. The black lines correspond to the line directed by $\va{u}$ but passing through only one point of $\Lambda$. For these states $\Delta \hat H_i=\Delta \hat H_\text{coll}=0$, and we have by convention defined $F=0$, thus not leading to acceptable optimal states. The red dashed lines are also directed by the vector $\va{u}$ but are passing through at least two points of $\Lambda$.}
    \label{fig:not evenly spaced grid}
\end{figure}

\paragraph{For the simplified bound}\label{parag: equality for simplified bound}~\\
If we still denote by $h_{\max}$ and $h_{\min}$ the extremal eigenvalues of $\hat H$, the states saturating the bound~(\ref{eq: general bound in finite dim}) are given by
\begin{equation}\label{eq: saturate simplified eq}
    \ket{\psi}=\frac{1}{\sqrt{2}}\left[\ket{\psi_{\min}}^{\otimes n}+e^{i\varphi}\ket{\psi_{\max}}^{\otimes n}\right],
\end{equation}
where $\ket{\psi_{\min}}$ and $\ket{\psi_{\max}}$ are the eigenvectors of $\hat H$ corresponding the the minimal and maximal eigenvalues and the phase $e^{i\varphi}$ can be arbitrary. Details are provided in appendix~\ref{appendix popoviciu's}. Up to the additional phase $e^{i\varphi}$, states in Eq.~(\ref{eq: saturate simplified eq}) correspond to Greenberger-Horne-Zeilinger-like states (GHZ) \cite{greenberger_going_1989}, which were previously discussed in \cite{pezze_quantum_2014}. In that work, it was shown that these states saturate the bound in Eq.~(\ref{eq: general bound in finite dim}). In this paper, we provide a complete characterization of all such optimal states. GHZ states, as well as their extension to qudits\footnote{For a $d$ dimensional Hilbert space with basis $\ket{1},\dots\ket{d}$ GHZ-like states are generally defined as $\ket{\text{GHZ}_d}=\sum_{i=1}^d \ket{i\cdots i}/\sqrt{d}$. They correspond to a subclass of the discrete version of the states in Eq.~(\ref{eq: first diagonal states}) with constant coefficients $C_\lambda=1/\sqrt{d}$ for which we have shown that they maximize $F$.}, have numerous applications in quantum metrology \cite{toth_quantum_2014}, quantum communication \cite{durt_security_2004,cerf_security_2002}, and quantum computing \cite{raussendorf_one-way_2001}. Our proposed quantifier could be used to assess the quality of the states produced, and imperfections in the implementation of related protocols could be quantified in terms of it. We leave these investigations for future work. \\

It is important to notice that the state of Eq.~(\ref{eq: saturate simplified eq}) also saturates the original inequality~(\ref{eq: general bound collective global}). However, as already pointed out, a much richer class of states saturates Eq.~(\ref{eq: general bound collective global}).

\subsubsection{Time frequency case}\label{subsubsec: for time-frequency}
Time-frequency states saturating equality~(\ref{eq: general bound collective global}) have already been analyzed in \cite{descamps_quantum_2023}. However, in the following, we want to discuss them in light of the more general picture we provide here. As discussed in Section~\ref{subsubsec: geometrical consequence}, the spectrum of the local operator plays a huge role in the determination of the states that saturate this bound. Thus, it should be noted that the spectrum of $\hat\omega$ is the whole real line $\Lambda=\R$.\\

As the analytical derivation has already been done in \cite{descamps_quantum_2023}, in the present paper we employ only the geometrical picture. As stated previously, states saturating the bound have to satisfy $\Delta(\hat P_j)=0$ for $j\geq 2$, which means that when plotted as functions on $\R^n$, the spectrum $S(\omega_1,\dots,\omega_n)$ should have its support lying on a line, directed by the vector $\va{u}=(1,\dots,1)$. Although in the case of finite dimension, there was only a discrete set of points on which $S$ could lie, which restricted the possibilities, in the present situation, the entire volume of $\R^n$ is accessible. This is pictured in two dimensions in Fig.~\ref{fig:saturation continuous}. A consequence of this is that all functions $S$ with concentrated support along the collective direction $\va{u}$ define a state that saturates the bound. From these observations, we arrive at the general formula describing such states
\begin{equation}
    \ket{\psi}=\int \dd\Omega S(\Omega)\ket{\Omega+\omega_1^0,\dots,\Omega+\omega_n^0}.
\end{equation}
The function $S$ encodes the way in which the spectrum of $\ket{\psi}$ is distributed (amplitude and phase) along the line, while the constants $\omega_i^0$ control the global shift of this line in $\R^n$. It is important to recognize that such a state is not physical, as it requires perfect correlation between the frequencies of the different photons. This is mathematically reflected by the fact that integration is performed only on one variable and thus they are Dirac delta distributions hiding in the spectrum $S$ of $\ket{\psi}$. These states, therefore, are not normalizable and, strictly speaking, do not belong to the Hilbert space (for example, like the ``generalized eigenstate" of the position or momentum operators). However, as discussed in our previous work, these states play a significant role in metrology \cite{descamps_quantum_2023,descamps_time-frequency_2023} and may also have potential applications in error correction \cite{descamps_gottesman-kitaev-preskill_2024}.
\begin{figure}
    \centering
    \scalebox{0.6}{\begin{tikzpicture}
	\begin{pgfonlayer}{nodelayer}
		\node [style=none] (10) at (0, 17.5) {};
		\node [style=none] (11) at (0, -1) {};
		\node [style=none] (12) at (-1, 0) {};
		\node [style=none] (13) at (17.5, 0) {};
		\node [style=none] (68) at (4, 8) {};
		\node [style=none] (69) at (6, 10) {};
		\node [style=none] (70) at (4.5, 9.25) {$\va{u}$};
		\node [style=none] (71) at (6, 2) {};
		\node [style=none] (72) at (9, 5) {};
		\node [style=none] (73) at (12, 3) {};
		\node [style=none] (74) at (16, 7) {};
		\node [style=none] (75) at (9, 11) {};
		\node [style=none] (76) at (13, 15) {};
		\node [style=none] (77) at (1, 1) {};
		\node [style=none] (78) at (15, 15) {};
		\node [style=none] (79) at (1, 9) {};
		\node [style=none] (80) at (2, 10) {};
		\node [style=none] (81) at (2, 14) {};
		\node [style=none] (82) at (5, 17) {};
		\node [style=none] (83) at (11, 7) {};
		\node [style=none] (84) at (12, 8) {};
		\node [style=none] (85) at (13, 9) {};
		\node [style=none] (86) at (14, 10) {};
		\node [style=none] (87) at (6, 14) {};
		\node [style=none] (88) at (9, 17) {};
	\end{pgfonlayer}
	\begin{pgfonlayer}{edgelayer}
		\draw [style=Thick arrow] (11.center) to (10.center);
		\draw [style=Thick arrow] (12.center) to (13.center);
		\draw [style=Thick arrow] (68.center) to (69.center);
		\draw [style=red line] (71.center) to (72.center);
		\draw [style=red line] (73.center) to (74.center);
		\draw [style=red line] (75.center) to (76.center);
		\draw [style=Blue line] (78.center) to (77.center);
		\draw [style=red line] (79.center) to (80.center);
		\draw [style=red line] (81.center) to (82.center);
		\draw [style=red line] (83.center) to (84.center);
		\draw [style=red line] (85.center) to (86.center);
		\draw [style=dashes red] (72.center) to (83.center);
		\draw [style=dashes red] (84.center) to (85.center);
		\draw [style=red line] (87.center) to (88.center);
		\draw [style=dashes red] (80.center) to (87.center);
	\end{pgfonlayer}
\end{tikzpicture}}
    \caption{Spectral space of $\mathcal S_2$. The colored lines represent possible support of states saturating the inequality~(\ref{eq: general bound collective global}). The blue lines correspond to the main diagonal. The red ones indicate other possibilities. The dotted lines emphasize the fact that the support of the spectrum does not have to be connected.}
    \label{fig:saturation continuous}
\end{figure}

\subsection{Adding thickness}\label{subsec: adding thickness}
For experimental reasons, the states saturating the bound~(\ref{eq: general bound collective global}) may not be accessible. As seen before, it can be because such states are not physical: infinite multi-mode squeezing for the quadrature \cite{fabre_modes_2020} or infinitely thin spectrum for time-frequency states. These states are non-normalizable and, as a result, cannot be experimentally produced, as they do not belong to the Hilbert space. It can also be due to the experimental limitation with the inability to generate states that are perfectly entangled along the considered variables. Thus, we may want to understand how the bound Eq.~(\ref{eq: general bound collective global}) is modified in these situations. We mathematically model this through the following inequalities 
\begin{equation}
    \forall j\geq 2,\;\Delta  \hat P_j\geq \zeta \Delta \hat P_1,
\end{equation}
where the operators $\hat P_j$ were introduced in the proof of Eq.~(\ref{eq: general bound collective global}). The positive real parameter $\zeta$ controls how thin the states can be. Although these operators were initially introduced as mathematical tools for the derivation, the choice of definition for $\hat P_j$ is significant in this context. Specifically, the condition $\Delta  \hat P_j\geq \zeta \Delta \hat P_1$ is not invariant under a change of the operator $\hat P_j$. The spectral representation provides an intuitive understanding of the non-zero thickness condition. Fig.~\ref{fig: non-zero thickness example} shows the spectral representation of such states for the finite-dimensional case and time-frequency systems when $n=2$. While we allow the states to be predominantly elongated in the direction of the main collective variable, the representation must retain some thickness in the orthogonal directions. Adapting the last step of the former proof (see Eq.~\ref{eq: last step proof inequ}), one gets that
\begin{equation}
    \max_i\Delta(\hat H_i)\geq \frac{1}{n}\Big(1+\zeta (n-1)\Big)\Delta(\hat P_1),
\end{equation}
which can be rearranged as
\begin{equation}
    F_{\ket{\psi}}\leq \frac{n^2}{(1-\zeta )+\zeta n}.
\end{equation}
A first analysis of this bound can be made by exploring different values of $\zeta$. For $\zeta=0$ (the constraint becomes trivial), we indeed recover Eq.~(\ref{eq: general bound collective global}). For $\zeta=1$, we get $F\leq n$. In this case, the state $\ket{\psi}$ has to be equally distributed in all directions. The last extreme value is $\zeta=\infty$, for which $F=0$. In this case, the state is infinitely thin toward the collective operator $\hat H_\text{coll}$.

\begin{figure*}
    \centering
    \scalebox{1.1}{\begin{tikzpicture}
	\begin{pgfonlayer}{nodelayer}
		\node [style=none] (0) at (-1, 0) {};
		\node [style=none] (1) at (0, -1) {};
		\node [style=none] (2) at (0, 9.5) {};
		\node [style=none] (3) at (9.5, 0) {};
		\node [style=none] (4) at (1.5, -0.25) {};
		\node [style=none] (5) at (1.5, 0) {};
		\node [style=none] (6) at (3.5, -0.25) {};
		\node [style=none] (7) at (3.5, 0) {};
		\node [style=none] (8) at (4.5, -0.25) {};
		\node [style=none] (9) at (4.5, 0) {};
		\node [style=none] (10) at (8, -0.25) {};
		\node [style=none] (11) at (8, 0) {};
		\node [style=none] (12) at (1.5, -0.25) {};
		\node [style=none] (13) at (0, 1.5) {};
		\node [style=none] (14) at (-0.25, 1.5) {};
		\node [style=none] (15) at (0, 3.5) {};
		\node [style=none] (16) at (-0.25, 3.5) {};
		\node [style=none] (17) at (0, 4.5) {};
		\node [style=none] (18) at (-0.25, 4.5) {};
		\node [style=none] (19) at (0, 8) {};
		\node [style=none] (20) at (-0.25, 8) {};
		\node [style=none] (21) at (1.25, 8) {};
		\node [style=none] (22) at (1.75, 8) {};
		\node [style=none] (23) at (1.5, 8.25) {};
		\node [style=none] (24) at (1.5, 7.75) {};
		\node [style=none] (25) at (1.25, 4.5) {};
		\node [style=none] (26) at (1.75, 4.5) {};
		\node [style=none] (27) at (1.5, 4.75) {};
		\node [style=none] (28) at (1.5, 4.25) {};
		\node [style=none] (29) at (1.25, 3.5) {};
		\node [style=none] (30) at (1.75, 3.5) {};
		\node [style=none] (31) at (1.5, 3.75) {};
		\node [style=none] (32) at (1.5, 3.25) {};
		\node [style=none] (37) at (3.25, 8) {};
		\node [style=none] (38) at (3.75, 8) {};
		\node [style=none] (39) at (3.5, 8.25) {};
		\node [style=none] (40) at (3.5, 7.75) {};
		\node [style=none] (49) at (3.25, 1.5) {};
		\node [style=none] (50) at (3.75, 1.5) {};
		\node [style=none] (51) at (3.5, 1.75) {};
		\node [style=none] (52) at (3.5, 1.25) {};
		\node [style=none] (53) at (4.25, 8) {};
		\node [style=none] (54) at (4.75, 8) {};
		\node [style=none] (55) at (4.5, 8.25) {};
		\node [style=none] (56) at (4.5, 7.75) {};
		\node [style=none] (57) at (4.25, 4.5) {};
		\node [style=none] (58) at (4.75, 4.5) {};
		\node [style=none] (59) at (4.5, 4.75) {};
		\node [style=none] (60) at (4.5, 4.25) {};
		\node [style=none] (61) at (4.25, 3.5) {};
		\node [style=none] (62) at (4.75, 3.5) {};
		\node [style=none] (63) at (4.5, 3.75) {};
		\node [style=none] (64) at (4.5, 3.25) {};
		\node [style=none] (65) at (4.25, 1.5) {};
		\node [style=none] (66) at (4.75, 1.5) {};
		\node [style=none] (67) at (4.5, 1.75) {};
		\node [style=none] (68) at (4.5, 1.25) {};
		\node [style=none] (73) at (7.75, 4.5) {};
		\node [style=none] (74) at (8.25, 4.5) {};
		\node [style=none] (75) at (8, 4.75) {};
		\node [style=none] (76) at (8, 4.25) {};
		\node [style=none] (77) at (7.75, 3.5) {};
		\node [style=none] (78) at (8.25, 3.5) {};
		\node [style=none] (79) at (8, 3.75) {};
		\node [style=none] (80) at (8, 3.25) {};
		\node [style=none] (81) at (7.75, 1.5) {};
		\node [style=none] (82) at (8.25, 1.5) {};
		\node [style=none] (83) at (8, 1.75) {};
		\node [style=none] (84) at (8, 1.25) {};
		\node [style=none] (85) at (0, 10.25) {$\lambda_2$};
		\node [style=none] (86) at (10.25, 0) {$\lambda_1$};
		\node [style=red dot] (87) at (1.5, 1.5) {};
		\node [style=red dot] (88) at (3.5, 3.5) {};
		\node [style=red dot] (89) at (3.5, 4.5) {};
		\node [style=red dot] (90) at (8, 8) {};
		\node [style=none] (94) at (5.25, 6.25) {};
		\node [style=none] (95) at (5.75, 5.75) {};
		\node [style=none] (96) at (6, 5.5) {};
		\node [style=none] (97) at (3.5, 6.5) {\scalebox{0.7}{$\Delta\hat P_2=\zeta\Delta\hat P_1$}};
		\node [style=none] (98) at (1.5, 0.75) {};
		\node [style=none] (99) at (3.5, 2.75) {};
		\node [style=none] (100) at (8.75, 8) {};
		\node [style=none] (101) at (6, 4.25) {\scalebox{0.7}{$\Delta \hat  P_1$}};
		\node [style=none] (102) at (11.25, 0) {};
		\node [style=none] (103) at (12.25, -1) {};
		\node [style=none] (104) at (12.25, 9.5) {};
		\node [style=none] (105) at (21.75, 0) {};
		\node [style=none] (171) at (12.25, 10.25) {$\omega_2$};
		\node [style=none] (172) at (22.5, 0) {$\omega_1$};
		\node [style=none] (173) at (14, 0.5) {};
		\node [style=none] (174) at (20, 6.5) {};
		\node [style=none] (175) at (16.75, 3.75) {};
		\node [style=none] (176) at (17.25, 3.25) {};
		\node [style=none] (177) at (19, 3.25) {\scalebox{0.7}{$\Delta\hat P_2=\zeta\Delta\hat P_1$}};
		\node [style=none] (178) at (18, 5.25) {};
		\node [style=none] (179) at (18.5, 4.75) {};
		\node [style=none] (180) at (18.75, 4.5) {};
		\node [style=none] (181) at (13.5, 1) {};
		\node [style=none] (182) at (15.5, 3) {};
		\node [style=none] (183) at (19.5, 7) {};
		\node [style=none] (184) at (15, 4.25) {\scalebox{0.7}{$\Delta \hat  P_1$}};
	\end{pgfonlayer}
	\begin{pgfonlayer}{edgelayer}
		\draw [style=Thick arrow] (0.center) to (3.center);
		\draw [style=Thick arrow] (1.center) to (2.center);
		\draw (5.center) to (12.center);
		\draw (7.center) to (6.center);
		\draw (9.center) to (8.center);
		\draw (11.center) to (10.center);
		\draw (20.center) to (19.center);
		\draw (18.center) to (17.center);
		\draw (16.center) to (15.center);
		\draw (14.center) to (13.center);
		\draw [style=Gray line] (23.center) to (24.center);
		\draw [style=Gray line] (21.center) to (22.center);
		\draw [style=Gray line] (27.center) to (28.center);
		\draw [style=Gray line] (25.center) to (26.center);
		\draw [style=Gray line] (31.center) to (32.center);
		\draw [style=Gray line] (29.center) to (30.center);
		\draw [style=Gray line] (39.center) to (40.center);
		\draw [style=Gray line] (37.center) to (38.center);
		\draw [style=Gray line] (51.center) to (52.center);
		\draw [style=Gray line] (49.center) to (50.center);
		\draw [style=Gray line] (55.center) to (56.center);
		\draw [style=Gray line] (53.center) to (54.center);
		\draw [style=Gray line] (59.center) to (60.center);
		\draw [style=Gray line] (57.center) to (58.center);
		\draw [style=Gray line] (63.center) to (64.center);
		\draw [style=Gray line] (61.center) to (62.center);
		\draw [style=Gray line] (67.center) to (68.center);
		\draw [style=Gray line] (65.center) to (66.center);
		\draw [style=Gray line] (75.center) to (76.center);
		\draw [style=Gray line] (73.center) to (74.center);
		\draw [style=Gray line] (79.center) to (80.center);
		\draw [style=Gray line] (77.center) to (78.center);
		\draw [style=Gray line] (83.center) to (84.center);
		\draw [style=Gray line] (81.center) to (82.center);
		\draw [style=Arrow] (95.center) to (96.center);
		\draw [style=Arrow] (95.center) to (94.center);
		\draw [style=Arrow] (99.center) to (100.center);
		\draw [style=Arrow] (99.center) to (98.center);
		\draw [style=Thick arrow] (102.center) to (105.center);
		\draw [style=Thick arrow] (103.center) to (104.center);
		\draw [style=Fill pink] (176.center)
			 to [in=-45, out=-135, looseness=0.25] (173.center)
			 to [in=-135, out=135, looseness=0.25] (175.center)
			 to [in=135, out=45, looseness=0.25] (174.center)
			 to [in=45, out=-45, looseness=0.25] cycle;
		\draw [style=Arrow] (179.center) to (180.center);
		\draw [style=Arrow] (179.center) to (178.center);
		\draw [style=Arrow] (182.center) to (183.center);
		\draw [style=Arrow] (182.center) to (181.center);
	\end{pgfonlayer}
\end{tikzpicture}}
    \caption{Two examples of states with non-zero width in the secondary variables in two dimensions. On the left: $\mathcal H$ is finite-dimensional. On the left panel, crosses represent the spectral space (all the possible couples $(\lambda_1,\lambda_2)$ of eigenvalues of $\hat H$), and the red filled circle represents the support of a particular state. On the right: $\mathcal H$ is the Hilbert space of two time-frequency single photon states. In both plots, the orthogonal arrows represent the geometric size of the state in spectral space, which is related to the variance of the corresponding operators.}
    \label{fig: non-zero thickness example}
\end{figure*}

\subsection{Extension to mixed states}\label{subsect: extension to mixed states}
The bound of Eq.~(\ref{eq: general bound collective global}) and the quantity $F$ can be extended to mixed states in multiple ways. As the only ingredient used in the derivation of Eq.~(\ref{eq: general bound collective global}) is the bilinearity and the positivity of the covariance, any construction based on this type of operation will immediately work. This is the case for the quantum variance naturally extended to the mixed state $\hat\rho$ as $\Delta_{\hat \rho}(\hat A)=\Tr(\hat A^2\hat \rho)-\Tr(\hat A\hat \rho)^2$ and the quantum Fisher information $\mathcal Q_{\hat\rho}(\hat A)$. Using the inequality $4\Delta_{\hat \rho}(\hat A)\geq  \mathcal Q_{\hat \rho}(\hat A)$ \cite{toth_extremal_2013,yu_quantum_2013}, we get
\begin{equation}
     \frac{\mathcal Q_{\hat \rho}(\hat H_\text{coll})}{4\max\limits_i \Delta_{\hat \rho}(\hat H_i)}\leq \frac{\mathcal Q_{\hat \rho}(\hat H_\text{coll})}{\max\limits_i \mathcal Q_{\hat \rho}(\hat H_i)}\leq n^2.
\end{equation}
As well as
\begin{equation}
    \frac{ \Delta_{\hat \rho}(\hat H_\text{coll})}{\max\limits_i \Delta_{\hat \rho}(\hat H_i)}\leq n^2.
\end{equation}
In the case of non-zero thickness, using the hypothesis that $\mathcal Q(\hat P_j)\geq \zeta  \mathcal Q(\hat P_1)$ (for $j\geq 2$) or $ \Delta(\hat P_j)\geq \zeta  \Delta(\hat P_1)$ (for $j\geq 2$) allows to modify the right-hand side in a similar manner $n^2\mapsto  \frac{n^2}{(1-\zeta )+\zeta n}$.\\

The steps we follow to extend $F$ to mixed states are similar to those outlined in \cite{chen_wigner-yanase_2005}. As explained in the Discussion Section~\ref{subsec: motivation}, their work begins by observing that for pure states of qubit systems, the variance of the collective spin observable $\hat{H}_\text{coll} = \hat{Z}_1 + \cdots + \hat{Z}_n$ satisfies the inequality $\Delta_{\ket{\psi}} \hat{H}_\text{coll} \leq n^2$. The extension to product states is then achieved by seeking a convex generalization of the variance to mixed states. One such extension is the quantum Fisher information, which is proportional to the convex roof of the variance \cite{yu_quantum_2013}. We adopt this approach in the following. An alternative is the Wigner-Yanase skew information \cite{wigner_information_1963}, which is a convex function that reduces to the variance for pure states. Although this choice is less optimal for detecting entanglement, since it is smaller than the quantum Fisher information, it has the advantage of being computationally simpler.\\

In the following, we would like to obtain a formula that is convex in the quantum states. It is well known that the quantum Fisher information is convex \cite{yu_quantum_2013}. However, due to the presence of the denominator, it is not clear that any of the previously proposed expressions are convex. To remedy this, we propose three solutions.

\subsubsection{Convex roof approach}\label{subsubsec: convex roof}
The first approach consists of forcing the convexity, more specifically, by considering the largest convex extension. It can be theoretically described by the convex roof construction \cite{toth_extremal_2013,yu_quantum_2013,uhlmann_roofs_2010}. More specifically, it is defined by the following formula
\begin{align}\label{eq: convex roof def of F}
    F^\text{CR}_{\hat \rho}&=\inf_{\{p_i,\ket{\psi_i}\}} \sum_i p_i F_{\ket{\psi_i}}\notag\\
    &=\inf_{\{p_i,\ket{\psi_i}\}} \sum_i p_i\frac{\Delta_{\ket{\psi_i}}(\hat H_\text{coll})}{\max_j\Delta_{\ket{\psi_i}}(\hat H_j)},
\end{align}
where the infimum is taken over all decompositions of the mixed state $\hat\rho=\sum p_i\ketbra{\psi_i}$ ($p_i>0$). Since the state $\ket{\psi_i}$ appearing in the decomposition of $\ket{\psi}$ may satisfy $\max_j\Delta_{\ket{\psi_i}}(\hat H_j)=0$, it is important to adopt the convention that $F=0$ for such state. We can verify that this quantity
\begin{itemize}
    \item is less than $n^2$,
    \item reduces to the previous definition for pure states (quotient of the global variance by the maximum of local ones),
    \item is a convex function of the states, and 
    \item is the largest function satisfying the two above properties.
\end{itemize}
The first property follows directly from the fact that it holds for pure states, and taking a convex combination cannot increase its value. The remaining three properties are derived from general principles of the convex roof construction (see Appendix~\ref{appendix: convex roof} for detailed explanations and proofs). This construction, commonly referred to as the convex roof construction, is a powerful method for extending functions defined on pure states to mixed states in a convex manner. It holds significant theoretical importance, as it yields the largest convex extension of the original function. However, despite its elegance, the abstract nature of this definition makes the quantity difficult to interpret intuitively. More importantly, from a practical standpoint, it is generally unclear whether this quantity can be efficiently computed.

\subsubsection{Utilizing the support of $\hat \rho$}\label{subsubsec: support}
As discussed before the difficulty of obtaining a quantity that is convex in the mixed state $\hat \rho$, is the presence of the denominator. The following proposition uses the convexity of the quantum Fisher information without breaking the convexity of the quotient by a clever choice of denominator.
\begin{equation}\label{eq: def F supp}
    F^\text{S}_{\hat \rho}=\frac{\mathcal Q_{\hat \rho}(\hat H_\text{coll})}{4\sup\limits_{\ket{\psi}\in\operatorname{Supp}(\hat\rho)} \max_i \Delta_{\ket{\psi}}(\hat H_i)}.
\end{equation}
Where $\operatorname{Supp}(\hat \rho)$ is the support of $\hat \rho$. The factor $4$ comes from the proportionality between the Quantum Fisher information and the variance for pure states: $\mathcal Q_{\ket{\psi}}(\hat H_\text{coll})=4\Delta_{\ket{\psi}}(\hat H_\text{coll})$. There are two equivalent definitions for the support of $\hat \rho$ in the case $\hat \rho$ has finite rank (details in appendix \ref{appendix: def support}).
\begin{itemize}
    \item $\ket{\psi}\in(\ker\hat\rho)^\perp$
    \item $\ket{\psi}$ appears in a convex decomposition of $\hat\rho$: $\hat\rho=\alpha\ketbra{\psi}+\sum_i\alpha_i\ketbra{\psi_i}$, for some $\alpha,\alpha_i>0$ and some pure states $\ket{\psi_i}$.
\end{itemize}

We can verify that $F^\text{S}_{\hat\rho}$ is indeed convex in the mixed state $\hat\rho$ (details are provided in appendix~\ref{appendix: formula support conv}). For finite rank mixed state, this definition is very useful since it is general, and more computable than the one of Eq.~(\ref{eq: convex roof def of F}). Indeed, in this case, the support of $\hat\rho$ is finite-dimensional, and thus the optimization is done over a compact set (the corresponding sphere of unit norm states) which ensures that the denominator is finite. However, if the rank is infinite, it is not clear whether $\sup\limits_{\ket{\psi}\in\operatorname{Supp}(\hat\rho)}$ is finite or not. In the case it is infinite, the quantity $F^\text{S}_{\hat\rho}$ is zero and loses all usefulness.

\subsubsection{Bounding the denominator}\label{subsubsec: bounding the denominator}
The last proposition is less general and is based on a similar idea to the last one. We use the numerator with the quantum Fisher information to get the convexity and try to bound the denominator, such that it does not have an impact on the convexity. If we restrict ourselves to mixed states $\hat \rho$, with bounded variance
\begin{equation}
    \max_i\Delta_{\hat \rho}(\hat H_i)\leq A,
\end{equation}
for a constant $A$ independent of the state $\hat \rho$, or such a class of states, we can define 
\begin{equation}
    F^\text{R}_{\hat\rho}=\frac{\mathcal Q_{\hat\rho}(\hat H_\text{coll})}{4A},
\end{equation}
which is convex, as it is proportional to the quantum Fisher information. This formula has the advantage of being simpler than the previous one, and thus more computable. However, we potentially need to restrict ourselves to a very small class of states. Furthermore, as the denominator is replaced by another quantity, the corresponding quantifier contains less information about the state. This construction is particularly relevant for finite-dimensional Hilbert space where the constant $A$ can be expressed in terms of the extremal eigenvalues of $\hat H$.
\begin{equation}
    F^\text{R}_{\hat \rho}=\frac{\mathcal Q_{\hat \rho}(\hat H_\text{coll})}{(h_{\max}-h_{\min})^2},
\end{equation}
satisfying $F_{\hat\rho}^\text{R}\leq n^2$ (remember the difference by a factor $4$ between the variance and the quantum Fisher information).

\subsubsection{Discussion on the extensions}
In this section, we briefly compare the advantages and drawbacks of the various extensions of $F$ to mixed states that were previously proposed. It is important to note that these extensions are presented in order of decreasing computational complexity, with $F^\text{CR}$ being the most computationally demanding and $F^\text{R}$ the least. On the other hand, we also have
\begin{equation}
    F^\text{R}\leq F^\text{S}\leq F^\text{CR}.
\end{equation}
The left inequality follows from the fact that if each term in the support of $\hat \rho$ satisfies $\max_i\Delta_{\ket{\psi_i}}(\hat H_i)\leq A$, then we also have $\sup\limits_{\ket{\psi}\in\operatorname{Supp}(\hat\rho)} \max_i \Delta_{\ket{\psi}}(\hat H_i)\leq A$. The right inequality, on the other hand, arises from the fact that $F^\text{CR}$ is the maximal convex extension of $F$ that coincides with $F$ on pure states. Since $F$ can serve as a witness of entanglement by violating the inequality $F\leq n$, larger extensions to mixed states are more effective in detecting entanglement. Therefore, $F^\text{CR}$ is the extension that provides the most information about the entanglement of a given state.

To better understand the merit and usefulness of the different extensions, we investigate how noise affects the amount of entanglement estimated by our quantifier. Since the extensions are computationally challenging, we limit our study to a simple noise toy model and leave more detailed explorations for future work. We begin with an $n$-qubit state $\ket{\psi}$, for which $F_{\ket{\psi}} = n^2$. For a noise parameter $\epsilon \in [0, 1]$, we consider two noisy versions of the state: $\hat \rho_\epsilon = (1-\epsilon)\ketbra{\psi}{\psi} + \epsilon \frac{\hat{\1}}{2^n}$ and $\hat \sigma_\epsilon = (1-\epsilon)\ketbra{\psi}{\psi} + \epsilon \ketbra{\varphi}{\varphi}$, where $\ket{\varphi}$ is a state with $F_{\ket{\varphi}} = 0$. The state $\hat \rho_\epsilon$ is obtained by applying a depolarizing channel to $\ket{\psi}$, while $\hat \sigma_\epsilon$ is formed by mixing the state $\ket{\psi}$ with $\ket{\varphi}$. Although the second noisy state is less physically motivated, it represents the simplest mathematical scenario where an optimal state is mixed with a state having low $F$.

The definition and computation details are provided in Appendix \ref{appendix: mixed computation}. Although we are unable to compute $F_{\hat \rho_\epsilon}^\text{CR}$, we can obtain the following
\begin{subequations}\label{eq: mixed state computation}
    \begin{equation}
        F_{\hat \rho_\epsilon}^\text{S}=F_{\hat \rho_\epsilon}^\text{R}=\frac{(1-\epsilon)^2}{(1-\epsilon)+\epsilon/2^{n-1}}n^2
    \end{equation}
    \begin{equation}
        F_{\hat \sigma_\epsilon}^\text{CR}=F_{\hat \sigma_\epsilon}^\text{S}=F_{\hat \sigma_\epsilon}^\text{R}=(1-\epsilon)n^2
    \end{equation}
\end{subequations}
We observe that these formulas continuously transition from the maximal value of $n^2$ (for $\epsilon=0$) to a minimum of 0 (for $\epsilon$), confirming that the noise reduces the amount of entanglement. Based on these expressions, it may seem that there is no difference between the various extensions of $F$ to mixed states. However, this is merely an artifact resulting from the choice of systems used in the computation. Specifically, the computation was performed for qubit systems, where the local variance is bounded by 1. Since the primary focus of this work is not on explicit computations, we leave further exploration to future studies.

\section{$k$-entanglement}\label{sec: k sep}

\subsection{The original statement in a general setting}\label{subsec: original statement in a general setting}
Following ideas developed in \cite{guhne_entanglement_2009,hyllus_fisher_2012,chen_wigner-yanase_2005,pezze_quantum_2014}, we use the formula obtained previously to analyze the notion of $k$-entanglement. We recall that one can consider a finer definition of the entanglement properties of states living in $\mathcal H ^{\otimes n}$ going beyond the separable/entangled dichotomy. The concept of $k$-entanglement is based on the following intuitive idea: for a set of $n$ parties, how many distinct groups can we form such that there are no quantum correlations or entanglement between the groups? Since entanglement is a valuable resource in many multipartite protocols - such as those used in sensing, computing, and communication - the study of $k$-entanglement is crucial. It helps us identify situations where the amount of entanglement in a given resource state is insufficient for the protocol to succeed. Additionally, it can guide the design of protocols that leverage less entangled states as resources, which are easier to produce. Formally, we say that a pure state $\ket{\psi}$ is $k$-separable (for integer $1\leq k\leq n$) if it can be written as a product of states that are at most entangled over $k$ copies of $\mathcal H$. This means that we can write (after permuting the systems)
\begin{equation}
    \ket{\psi}=\ket{\phi_1}\otimes \cdots\otimes \ket{\phi_l},
\end{equation}
where each state $\ket{\phi_i}\in\mathcal H^{\otimes r_i}$ with $r_i\leq k$. We say that a state is $k$-entangled if it is $k$-separable but not $(k-1)$-separable. The definitions can be extended to mixed states easily: $\hat \rho$ is $k$-separable if it is the mixture of pure $k$-separable states:
\begin{equation}
    \hat \rho =\sum_i p_i\ketbra{\psi_i},
\end{equation}
where $\ket{\psi_i}$ is $k$-separable. As for pure states, a mixed state is $k$-entangled if it is $k$-separable but not $(k-1)$-separable.\\

Following the proofs of \cite{hyllus_fisher_2012,chen_wigner-yanase_2005, pezze_quantum_2014}, using the knowledge that all $k$-partite pure states $\ket{\psi}$ satisfy $F_{\ket{\psi}}\leq k^2$, it follows that any $k$-separable pure state $\ket{\psi}$ has to satisfy 
\begin{equation}\label{eq: bound k entanglement}
    F_{\ket{\psi}}\leq \left\lfloor \frac{n}{k}\right\rfloor k^2+\left(n-\left\lfloor\frac{n}{k}\right\rfloor k\right)^2\leq nk.
\end{equation}
The left inequality is identical to the one derived for the Quantum Fisher information \cite{hyllus_fisher_2012, pezze_quantum_2014} and the Wigner-Yanase skew information \cite{chen_wigner-yanase_2005}. The right inequality provides an upper bound for the complex term by using a simpler expression that does not involve integer parts. To the best of our knowledge, and surprisingly given its simplicity, this is a new inequality that yields the more concise bound $F\leq nk$ for $k$-separable states. Although this bound is slightly less optimal, equality holds between both upper bounds when $n/k$ is an integer. As this statement can be made even more general, the proof is delayed for later. This leads to a similar $k$-entanglement criterion: if the inequality is violated for some value of $k$, then the state is at least $(k+1)$-entangled. Once again, similarly to the original situation, the statement can be extended to mixed states, provided the extension used is convex.

\subsection{The general statement}\label{subsec: k sep general statement}
The phenomenon observed previously is even more general. It can be extended as follows: we still consider a Hilbert space  $\mathcal H$ and we assume that we have two maps $F$ and $f$ defined on 
\begin{equation}
    \begin{aligned}
        F&:\bigcup_{n=1}^{+\infty} \mathcal D(\mathcal H^{\otimes n}) \to \R,\\
        f&:\R\to \R,
    \end{aligned}
\end{equation}
where $\mathcal D(\cdot)$ denotes the set of mixed states of the corresponding Hilbert space. If we assume that
\begin{itemize}
    \item For all integer $n$, $F$ is convex on $\mathcal D(\mathcal H^{\otimes n})$.
    \item $F$ is sub-additive on pure product states : $F(\ketbra{\psi}\otimes\ketbra{\phi})\leq F(\ketbra{\psi})+F(\ketbra{\phi})$.
    \item The function $f$ is convex and $f(0)=0$.
    \item For all integer $n$, there exist a subset $\Gamma_n\subset\mathcal D(\mathcal H^{\otimes n})$ on which $F(\hat \rho)\leq f(n)$.
\end{itemize}
We further define the notion of $k$-producibility and $k$-entanglement over $\Gamma=\bigcup \Gamma_n$, simply by adapting the previous definition by asking that all states appearing in the decompositions are in some $\Gamma_n$. This leads to the following general statement. If $\hat \rho\in \mathcal D(\mathcal H^{\otimes n})$ is $k$-separable quantum state over $\Gamma$ then
\begin{equation}\label{eq: k sep general ineq}
    F(\hat \rho)\leq \left\lfloor \frac{n}{k}\right\rfloor f(k)+f\left(n-\left\lfloor \frac{n}{k}\right\rfloor k\right)\leq \frac{n}{k}f(k).
\end{equation}

See appendix \ref{appendix: k sep} for a detailed proof. The sets $\Gamma_n$ correspond to either the set where $F$ is defined or where the condition $F(\hat \rho)\leq f(n)$ is verified. It can be a subclass of state where this constraint is stricter. This is the case for example, for states described in Sec.~\ref{subsec: adding thickness} where a constraint on the thickness of the states was added. The right inequality generalizes the corresponding inequality in Eq.~(\ref{eq: bound k entanglement}), providing a simplified but less optimal upper bound, $F\leq \frac{n}{k}f(k)$. Naturally, when $k$ divides $n$, both expressions coincide.

Regardless of the extension to mixed states that we consider, the quantifier $F$ defined at the beginning of the paper satisfies both convexity and sub-additivity (see Eq.~\ref{eq: subadditivity}). Now, considering $\Gamma_n=\mathcal D(\mathcal H^{\otimes n})$ and the function $f:x\mapsto x^2$, the other two hypotheses are also satisfied. Our general theorem then reduces to the inequalities in (\ref{eq: bound k entanglement}). The flexibility of our formulation enables us to focus on a smaller subset of states, $\Gamma_n$, for which stricter inequalities, such as $F(\hat \rho)\leq f(n)$, hold for other functions $f$. In the next section, we will explore the consequences of these assumptions in the context of non-zero thickness.

The general approach for detecting $k$-entanglement has already been explored in \cite{hong_measure_2012}. Building on generalizations of concurrence\cite{wootters_entanglement_1998}, the authors introduced a family of quantifiers, denoted $C_{k-\text{ME}}$, which depend on the parameter $k$. These quantifiers effectively detect $k$-separable states. Therefore, to determine the level of $k$-entanglement for an unknown state, one must compute several such quantifiers. In contrast, our method relies on the calculation of a single quantifier, with the value of $F$ providing partial information about the entanglement level. Moreover, the definition of $C_{k-\text{ME}}$ involves considering all possible partitions of the $n$ systems into groups of size $k$, which grows rapidly with $n$ and $k$, making the computation challenging even for pure states. In this context, our quantifier presents a computationally simpler alternative for detecting $k$-entanglement.

\subsection{$k$-entanglement and non-zero thickness}\label{subsec: k sep and non-zero thickness}
In this section, we mix the two ingredients of $k$-entanglement and non-zero thickness. As discussed in the previous sections, both low values of $k$-entanglement and high thickness in the orthogonal direction (i.e., high values of $\zeta$) reduce the value of $F$. In realistic experimental scenarios, it may be challenging to produce states that are both highly entangled (with a large $k$) and have a thin spectrum (low $\zeta$). However, by relaxing these constraints and exploring the trade-off between these two parameters, we may identify experimentally feasible states that yield high values of $F$. If we consider a $k$-separable state $\ket{\psi}$, for which all the states appearing in the decomposition satisfy the non-zero thickness hypothesis $\Delta \hat P_j\geq \zeta  \Delta\hat P_1$, we have,
\begin{equation}\label{eq: ineq for k sep and non-zero thickness}
    \begin{aligned}
        F&\leq  \left\lfloor \frac{n}{k}\right\rfloor \frac{k^2}{(1-\zeta )+k\zeta }+\frac{\left(n-\left\lfloor\frac{n}{k}\right\rfloor k\right)^2}{(1-\zeta )+\left(n-\left\lfloor\frac{n}{k}\right\rfloor k\right)\zeta }\\
        &\leq  \frac{kn}{(1-\zeta )+k\zeta }.
    \end{aligned}
\end{equation}
Indeed, this situation fits the general picture presented in the previous section, where $\Gamma_n$ are the states satisfying the width condition and $f:x\mapsto \frac{x^2}{(1-\zeta )+x\zeta }$ is convex and satisfy $f(0)=0$ (it is simple to verify that $f''\geq 0$).\\

The inequality~(\ref{eq: ineq for k sep and non-zero thickness}) is a witness of at least $k+1$-entanglement under the assumption of non-zero width given by $\zeta$. The previous discussion of non-zero thickness is particularly relevant for time-frequency states. Indeed, in finite dimension, all states are physical, and thus {\it a priori}  their production is experimentally feasible. Although the non-zero thickness constraint can model the inability to produce highly correlated entangled states among distant parties, there are no fundamental limits. In the case of time-frequency states, the situation is different. As we have discussed, perfectly correlated states saturating the bound are non-normalizable states, thus only approximation can be produced in the lab. The non-zero thickness parameter $\zeta$ is thus a measure of how close to the ideal case are the experimental capabilities.\\

Recall that for a $k$-separable mix state for which all pure states appearing in its decomposition satisfy $\Delta\hat P_j\geq \zeta \Delta \hat P_1$ for $j\geq 2$, we have the inequality
\begin{equation}\label{eq: inequality with non-zero thickness}
    F\leq \frac{kn}{(1-\zeta )+k\zeta }.
\end{equation}
As previously analyzed, a violation of this inequality is the signature of $k+1$-entanglement. Recall that the quantity $F$ is constructed from a metrological point of view. It can be understood as a measure of the metrological usefulness of the entanglement. Indeed, it is constructed as the quotient of the quantum Fisher information (which is the standard measure of metrological capability) and the local variance (which quantifies the local metrological resource). From this insight, we want to study further the right-hand side of inequality~(\ref{eq: inequality with non-zero thickness}). As an upper bound on the quantity $F$ constrained by both the width parameter $\zeta$ and the size of local systems $k$, it measures the potential equivalence between two sets of parameters. A first interesting study is the comparison between the situations $(k,\zeta=0)$ and $(k=n,\zeta)$. This corresponds on one hand to the case of a $k$-separable state that has no thickness constraints. On the other, this corresponds to a potentially fully entangled state which thickness is restricted by the parameter $\zeta$. Asking that these two situations witness the same metrological bound, we arrive at the equality
\begin{equation}
    kn=\frac{n^2}{(1-\zeta)+n\zeta}.
\end{equation}
This can be solved for either $k$ or $\zeta$ so that 
\begin{align}\label{eq: extreme tradeoff n zeta}
    k=\frac{n}{(1-\zeta)+ n\zeta}, && \zeta =\frac{n-k}{k(n-1)}.
\end{align}
\begin{figure*}
    \centering
    \begin{tabular}{c c}
       \includegraphics[width=0.4\linewidth]{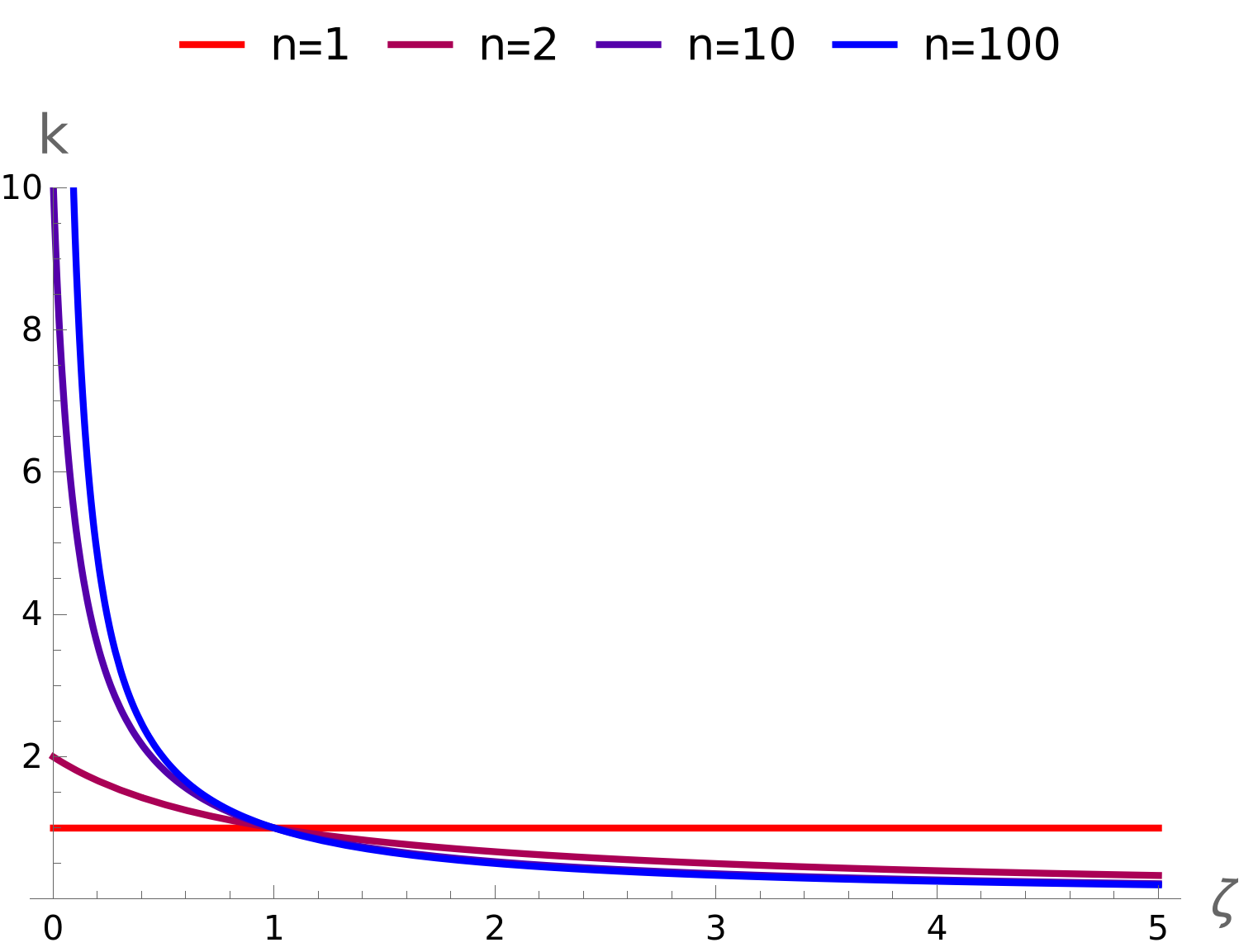}  &  \includegraphics[width=0.4\linewidth]{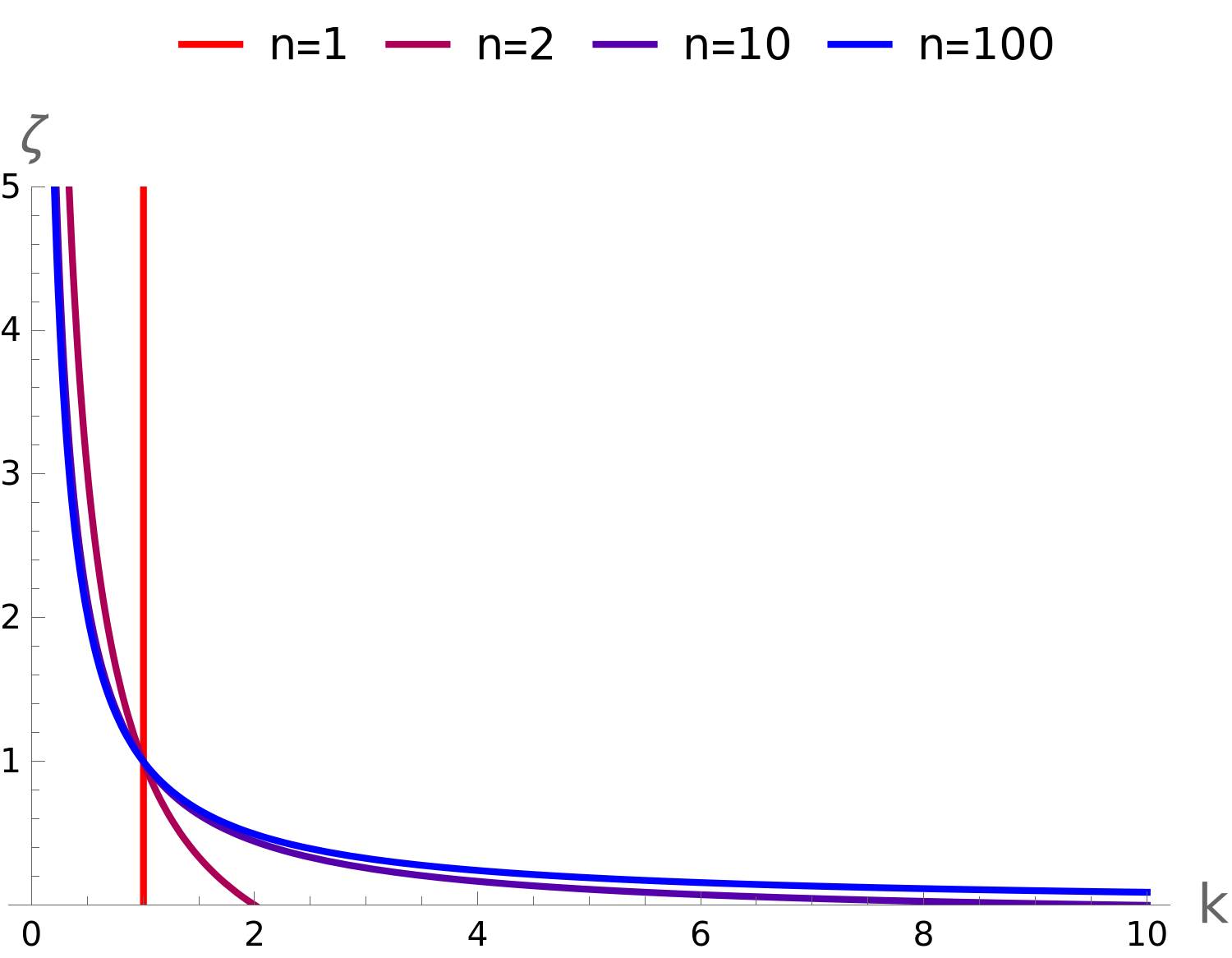}
    \end{tabular}
    \caption{Comparison of the values of $k$ and $\zeta$ such that a $k$-entangled state with no thickness $(k,\zeta=0)$ produces the same value of the upper bound of $F$ as a fully entangled state with thickness $\zeta$ $(k=n,\zeta)$ is shown. The left (right) plot expresses $k$ ($\zeta$) as a function of $\zeta$ ($k$) according to the left (right) equality in Eq.~(\ref{eq: extreme tradeoff n zeta}). The different curves correspond to different values of $n$. Since expressing one parameter in terms of the other, or vice versa, is an inversion, the two plots are related by a reflection across the line $y=x$. It is observed that all the curves are decreasing. This can be interpreted as follows: larger values of $\zeta$ (i.e., higher thickness) result in a smaller value of $F$, meaning that less $k$-entanglement is required to achieve the same value of $F$.}
    \label{fig: plot varying n}
\end{figure*}
It is interesting to verify these relations for the extreme cases. The case $k=n$ naturally corresponds to $\zeta=0$ as we are comparing the same situations. The other extreme is more interesting. Indeed, we can verify that $k=1$ corresponds to $\zeta=1$. In the first case, this corresponds to looking at a state which is $1$-separable {\it i.e.,} separable. Such a state necessarily exhibits the same thickness in all directions, which is exactly what is under the constraint $\Delta P_j = \Delta P_1$. In Fig. \ref{fig: plot varying n} we plot the relations of Eq. (\ref{eq: extreme tradeoff n zeta}).\\

More generally, for any value $f\in[0,n^2]$ we can solve for the parameters $(k,\zeta)$ that would lead to 
\begin{equation}
    \frac{kn}{(1-\zeta)+k\zeta}=f.
\end{equation}
Solving for $\zeta$ leads to $\zeta=\frac{kn-f}{f(k-1)}$. The condition $\zeta\geq 0$ imposes that a solution for $\zeta$ can exist if and only if $kn\geq f$. This makes intuitive sense as $kn$ is the largest value reachable by a $k$-separable state. Inversely, one can solve for $k$ and get $k=\frac{f(1-\zeta)}{n-f\zeta}$. This time $k$ must satisfy $k\leq n$ which happens when $f\leq \frac{n^2}{(1-\zeta)+n\zeta}$. Once again, this makes sense as this is the highest value one can reach with a non-zero thickness $\zeta$. We plot in Fig.~\ref{fig: plot varying f} both relations.\\
\begin{figure*}
    \centering
    \begin{tabular}{c c}
       \includegraphics[width=0.4\linewidth]{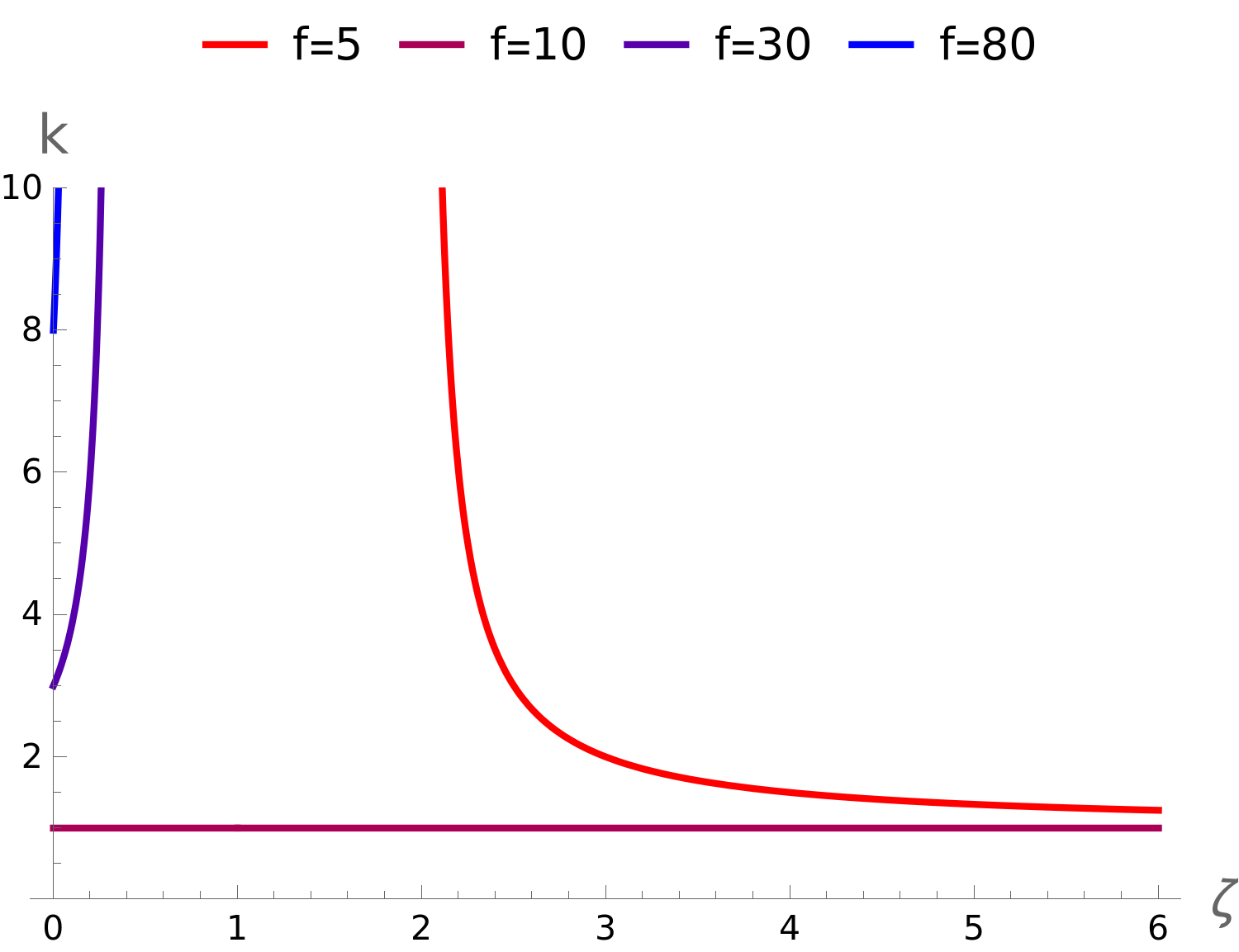}  &  \includegraphics[width=0.4\linewidth]{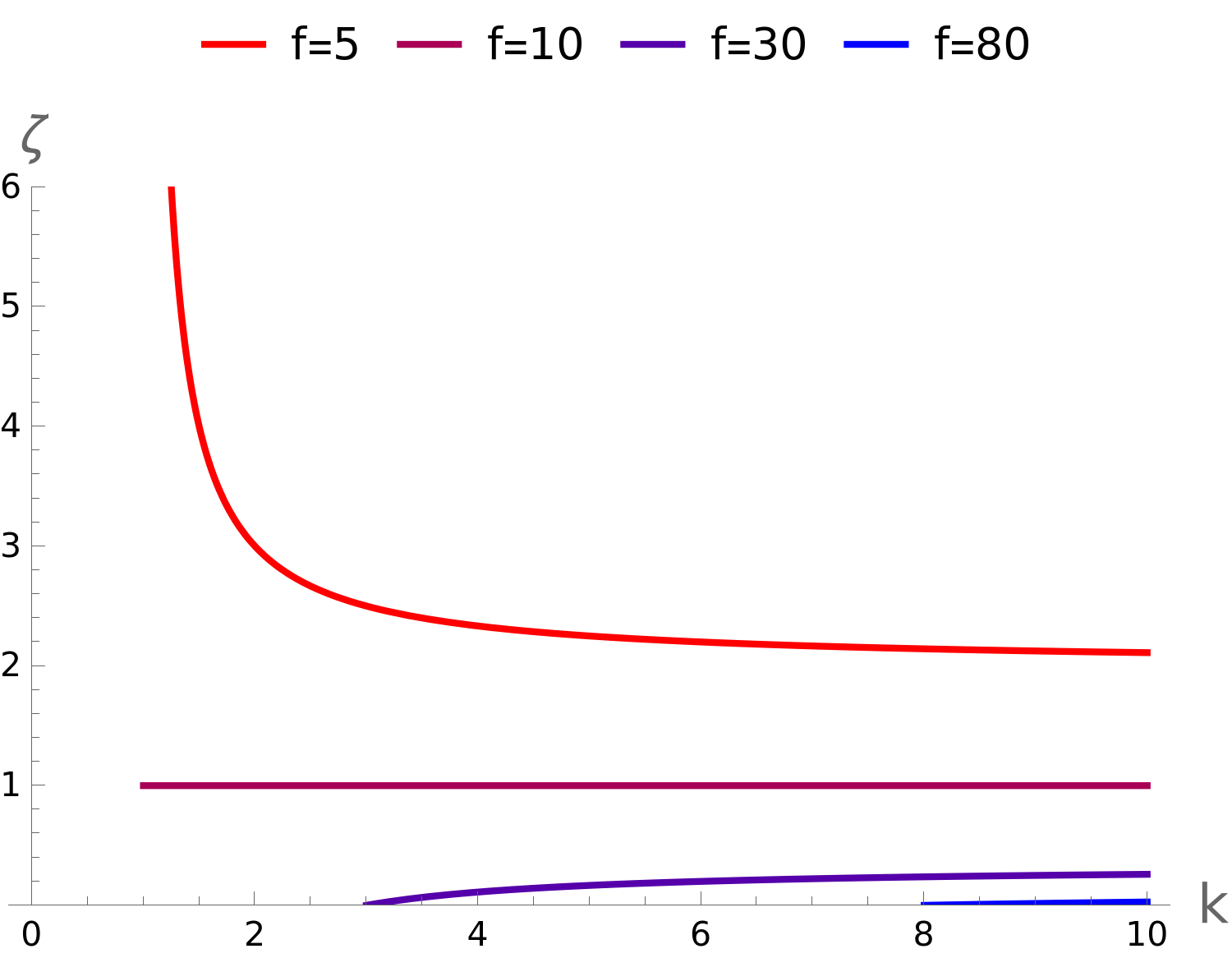}
    \end{tabular}
    \caption{Presentation of the required values of $k$ and $\zeta$ necessary to achieve the maximal value of $F=f$, plotted for different values of $f$. Both graphs are generated with $n=10$. The left (resp. right) plot shows, as a function of $\zeta$ (resp. $k$), the minimal (resp. maximal) value of $k$ (resp. $\zeta$) such that a state with parameters ($k$, $\zeta$) can achieve a value of $F$ at least $f$. Both plots are generally reflections of each other with respect to the line $y=x$. The only exception is the line $f=10=n$, which behaves differently. This can be understood by noting that for $f=n$, the condition $\frac{kn}{(1-\zeta)+\zeta k}=n$ is equivalent to $(\zeta-1)(k-1)=0$, which gives the solutions $k=1$ or $\zeta=1$. Intuitively, this means that if a state is fully separable ($k=1$) or exhibits rotational symmetry in all directions ($\zeta=1$), the value of the other parameter does not affect the quantity $F$. We observe that the curve for $f=5$ lies below the separable/rotationally symmetric limit $f=n$. This implies that a state can achieve $f=5$ with both large values of $\zeta$ and low values of $k$-entanglement. On the other hand, the curves for $f=30$ and $f=80$ demonstrate that high values of $\zeta$ and low values of $k$ are required to reach $F=f$. These observations are consistent with the inequalities derived from the conditions $k \leq n$ and $\zeta \geq 0$ discussed in the main text.}
    \label{fig: plot varying f}
\end{figure*}
All these computations show that when wanting to reach a certain level of entanglement in the collective variable, one can either optimize over the size $k$ of the local entanglement or reduce the width $\zeta$. In practical situations, there may be a trade-off between these two optimizations, and the computation above allows them to understand how to perform this optimization.

\section{Conclusion}\label{section: conclusion}
We have presented a general formulation of entanglement witnesses based on the variance of collective observables that applies to any Hilbert space. In addition, we have related our work to the notion of $k$-entanglement and studied different ways to compare it to entanglement resources. This work bridges the gap between numerous existing studies by providing a unified perspective. Our results are based on the introduction of a generalized entanglement quantifier, and we have discussed some of its properties as the type of entanglement it detects. The extension of the presented quantifier to mixed states and the exploration of all the possibilities and properties of the introduced quantifier require further analysis.

Another important contribution of the present manuscript is providing a pictorial geometrical representation of our results that helps the intuitive interpretation of the entanglement criteria and the introduction of the spectral space. Throughout the paper, this space allowed for a clear illustration of the idea developed and served as a visual aid for studying the quantifier $F$. This spectral space may be further studied in future works and may become a valuable tool for the development and presentation of new ideas. As collective variables have proven to be a useful tool for many different studies in the realm of quantum optics and quantum information such as improving metrological protocols or error correction schemes, the theoretical tools we provided in this study may allow for further progress in these research fields.

\newpage

\section{Acknowledgments}\label{section: acknoledgments}
We acknowledge funding from the Plan France 2030
through the project ANR-22-PETQ-0006.

\bibliographystyle{unsrturl} 
\bibliography{refs}

\appendix
\onecolumngrid

\section{Cauchy-Schwarz inequality for the general collective operator}\label{appendix: cauchy-schwarz general}
In Section~\ref{subsec: remark def coll op} we discuss the possibility of defining a collective operator via the more general formula
\begin{equation}
    \hat H_\text{coll}=\sum_{i=1}^n c_i \hat H_i,
\end{equation}
where the coefficient $c_i$ are any real numbers satisfying the normalization condition $\sum_{i=1}^n \abs{c_i}=n$. In this short appendix, we show that the condition $c_i=\pm 1$ is necessary so that the inequality
\begin{equation}
    \Delta(\hat H_\text{coll})\leq n^2\max_l\Delta \hat H_l,
\end{equation}
can have states saturating it. The idea is to look at the derivation of this inequality via the use of the Cauchy-Schwarz inequality. This is done as follows
\begin{subequations}
    \begin{align}
        \Delta\hat H_\text{coll}&=\sum_{i,j=1}^n \abs{c_ic_j}\operatorname{Cov}(\hat H_i,\hat H_j)\\
        &\leq\sum_{i,j=1}^n \abs{c_ic_j}\sqrt{\Delta\hat H_i\Delta \hat H_j}\\
        &\leq\left(\sum_{i=1}^n\abs{c_i}\right)^2\max_l\Delta \hat H_l\\
        &\leq \left(\sum_{i=1}^n\abs{c_i}^2\right)\left(\sum_{i=1}^n1\right)\max_l\Delta \hat H_l\\
        &=n^2\max_l\Delta \hat H_l,
    \end{align}
\end{subequations}
where the first inequality follows from the application of the Cauchy-Schwarz inequality to the quantum covariance. The second inequality is obtained by bounding $\Delta \hat H_i$ by the maximal local variance, while the third inequality is again derived using the Cauchy-Schwarz inequality, this time applied to the standard inner product on $\R^n$. The second use of the Cauchy-Schwarz inequality is independent of the states considered. Therefore, if this inequality is strict, no state can saturate inequality (\ref{eq: general bound collective global}). By analyzing the condition for equality in this step, we find that equation (\ref{eq: general bound collective global}) can only be saturated if $(\abs{c_1},\dots,\abs{c_n})\propto (1,\dots,1)$, which, given the normalization condition, implies that $c_i=\pm1$.

\section{Short on convex roof following Ref.~\cite{uhlmann_roofs_2010}}\label{appendix: convex roof}
Let's consider a convex set $\mathcal C$ and define its extreme points $\delta\mathcal{C}$ as
\begin{equation}
    \delta\mathcal C=\left\{x\in\mathcal C\;\middle|\; \forall y,z\in\mathcal C,\;\forall t\in(0,1),\; x=ty+(1-t)z\Rightarrow x=y=z\right\}.
\end{equation}
Intuitively, this corresponds to the points that cannot be written as a non-trivial convex combination. In the context of quantum mechanics, $\mathcal C$ is the set of all mixed states, while $\delta\mathcal C$ is the set of the pure states.  One can also comment, that {\it a priori} $\delta\mathcal C\neq \partial \mathcal C$ where $\partial \mathcal C$ is the topological boundary of $\mathcal C$. We further assume that $\delta\mathcal C$ generate the set $\mathcal C$
\begin{equation}
    \forall x\in\mathcal C,\; \exists x_1,\dots,x_n\in\delta\mathcal C\!:\; \exists p_1,\dots,p_n\geq0\!:\; \sum_{i=1}^n p_i=1\text{ and } x=\sum_{i=1}^n p_i x_i,
\end{equation}
meaning that each point of $\mathcal C$ can be written as a convex combination of extreme points. If $f: \delta\mathcal C\to \R$ is an arbitrary real function defined on the extreme points of $\mathcal C$, one can define 
\begin{align}
    F: \mathcal C &\to \R\\
    x&\mapsto \inf_{\{p_i|x_i\}}  \sum_i p_if(x_i),
\end{align}
where the infimum is taken over all convex decomposition of $x$: $x=\sum p_i x_i$ with $p_i\geq 0$. One can note that the hypothesis that $\mathcal C$ is generated by $\delta\mathcal C$ is not strictly necessary. With the convention that $\inf \emptyset=+\infty$, we still get a valid formula for $F$ that may take the value $+\infty$. Such construction satisfies the following properties
\begin{enumerate}
    \item For $x\in\delta\mathcal C$, $F(x)=f(x)$. On the extreme points, $F$ reduces to $f$.
    \item $F$ is convex.
    \item $F$ is the largest convex function that is equal to $f$ on $\delta\mathcal C$. This means that if $G$ is a convex function such that $\forall x\in\delta\mathcal C,\; G(x)=f(x)$ then $\forall x\in\mathcal C,\; G(x)\leq F(x)$ 
\end{enumerate}
We verify it as follows

\begin{proof}
Due to the definition of $\delta\mathcal C$, for any $x\in\delta\mathcal C$, there exists only one convex decomposition: the trivial one $x=x$. As such it follows that $F(x)=f(x)$.\\

$F$ is convex since, if we take $x,y\in\mathcal C$ and $t\in[0,1]$ then for any convex decomposition of $x=\sum p_i x_i$ and $y=\sum_j q_j y_j$
\begin{equation}
    t \sum_i p_i f(x_i) + (1-t)\sum_j q_j f(y_j)=  \sum_i tp_i f(x_i) + \sum_j (1-t)q_j f(y_j)\geq F( tx+(1-t)y),
\end{equation}
since $tx+(1-t)y=\sum tp_i x+\sum (1-t)q_j$ is a valid decomposition. By taking the infimum over all decomposition of $x$ and $y$, we indeed, have
\begin{equation}
    tF(x)+(1-t)F(y)\geq F(tx+(1-t)y),
\end{equation}
Thus $F$ is indeed convex.\\

Finally if $G$ is a convex function that equal $f$ on $\delta\mathcal C$, for any decomposition of $x\in\mathcal C$, $x=\sum_i p_i x_i$, with $x_i\in\delta\mathcal C$, we have
\begin{equation}
    G(x)\leq \sum_i p_i G(x_i)= \sum_i p_i f(x_i).
\end{equation}
The inequality comes from the convexity of $G$, while the equality comes from the fact that $G$ equals $f$ on $\delta\mathcal C$. By optimizing all possible such decompositions, we indeed get 
\begin{equation}
    G(x) \leq \inf_{\{p_i|x_i\}} \sum_i p_i f(x_i)=F(x).
\end{equation}
\end{proof}    

\section{Proof on support}\label{appendix: support}
\subsection{Two definitions of the support}\label{appendix: def support}
We first verify the equivalence between the two definitions of the support given in the main text.

\begin{proof}
\underline{$\Rightarrow$:} If $\ket{\psi}\in(\ker \hat\rho)^\perp$ then by definition for all state $\ket{\sigma}\in\ker\hat\rho$, we have $\bra{\psi}\ket{\sigma}=0$. Let us now consider a diagonalization of $\hat \rho$
\begin{equation}
    \hat \rho=\sum_i \lambda_i \ketbra{\psi_i}
\end{equation}
where $\lambda_i>0$ and the sum is finite since the rank of $\hat\rho$ is finite. We can then define $\lambda=\min_i \lambda_i$. Since there is only a finite number of non-zero eigenvalues, we know that $\lambda>0$. Then we claim that $\hat\rho-\lambda\ketbra{\psi}$ is a positive operator. Indeed, for any state $\ket{\phi}$ we verify that $\bra{\phi}(\hat\rho-\lambda\ketbra{\psi})\ket{\phi}\geq0$. To show this, we decompose $\ket{\phi}=\ket{\sigma}+\ket{\tau}$ where $\ket{\sigma}\in\ker\hat\rho$ and $\ket{\tau}\in(\ker\hat\rho)^\perp$. We get
\begin{subequations}
    \begin{align}
        \bra{\phi}(\hat\rho-\lambda\ketbra{\psi})\ket{\phi}&=\bra{\sigma}\hat\rho\ket{\sigma}-\lambda\bra{\sigma}\ket{\psi}\bra{\psi}\ket{\sigma}+\bra{\sigma}\hat\rho\ket{\tau}-\lambda\bra{\sigma}\ket{\psi}\bra{\psi}\ket{\tau}\notag\\
        &\qquad+\bra{\tau}\hat\rho\ket{\sigma}-\lambda\bra{\tau}\ket{\psi}\bra{\psi}\ket{\sigma}+\bra{\tau}\hat\rho\ket{\tau}-\lambda\bra{\tau}\ket{\psi}\bra{\psi}\ket{\tau}\\
        &=\bra{\tau}\hat\rho\ket{\tau}-\lambda\bra{\tau}\ket{\psi}\bra{\psi}\ket{\tau}
    \end{align}
\end{subequations}
since $\ket{\sigma}\in\ker\hat\rho$ thus $\hat\rho\ket{\sigma}=0$ and $\bra{\psi}\ket{\sigma}=0$. Since $\ket{\tau}\in(\ker\hat\rho)^\perp$ we can decompose it over the states $\ket{\psi_i}$: $\ket{\tau}=\sum \tau_i\ket{\psi_i}$. Thus
\begin{equation}
    \bra{\tau}\hat\rho\ket{\tau}=\sum_i \abs{\tau_i}^2\lambda_i\geq \lambda\sum_i \abs{\tau_i}^2=\lambda\bra{\tau}\ket{\tau}
\end{equation}
Finally, by Cauchy-Schwarz inequality, we have $\bra{\tau}\ket{\psi}\bra{\psi}\ket{\tau}\leq \braket{\tau}$. In the end, we Indeed, have
\begin{equation}
    \bra{\phi}(\hat\rho-\lambda\ketbra{\psi})\ket{\phi}\geq0
\end{equation}
To finish the proof, we simply have to say that since $\hat\rho-\lambda\ketbra{\psi}$ is a positive operator, we can diagonalize it and get $\hat\rho-\lambda\ketbra{\psi}=\sum_i\alpha_i\ketbra{\psi_i'}$ and we obtain
\begin{equation}
    \hat\rho =\lambda\ketbra{\psi}+\sum_i\alpha_i\ketbra{\psi_i'}
\end{equation}

\underline{$\Leftarrow$:} We now assume that, we can write $\hat\rho$ as a convex decomposition containing $\ket{\psi}$ 
\begin{equation}
    \hat\rho=\alpha\ketbra{\psi}+\sum_i\alpha_i\ketbra{\psi_i}
\end{equation}
If we consider $\ket{\phi}\in\ker\hat\rho$ we have 
\begin{equation}
    \bra{\phi}\hat\rho\ket{\phi}=\lambda\abs{\bra{\psi}\ket{\phi}}^2+\sum\alpha_i \abs{\bra{\psi_i}\ket{\phi}}^2=0
\end{equation}
Since we are only summing non-negative terms, this means that all the terms in the sum are zero, specifically we have $\bra{\psi}\ket{\phi}=0$. Since this is true for all states $\ket{\phi}$ in $\ker\hat\rho$, this means that $\ket{\psi}\in(\ker\hat\rho)^\perp$.
\end{proof}    

\subsection{The formula with support is convex}\label{appendix: formula support conv}
We now prove that the quantity $F^\text{S}$ defined in Eq.~(\ref{eq: def F supp}) of the main text is convex in the quantum states.
\begin{proof}
Let us consider two mixed state $\hat \rho_1$ and $\hat\rho_2$ and a convex combination $\hat\rho=t\hat\rho_1+(1-t)\hat\rho_2$. Since the quantum Fisher information is indeed, convex, we have
\begin{equation}
    \mathcal Q_{\hat\rho}(\hat H_\text{coll})\leq t\mathcal Q_{\hat\rho_1}(\hat H_\text{coll})+(1-t)\mathcal Q_{\hat\rho_2}(\hat H_\text{coll}).
\end{equation}
Moreover, due to the definition of the support in terms of convex decomposition, it is clear that 
\begin{equation}
    \operatorname{Supp}(\hat\rho_i)\subset \operatorname{Supp}(\hat\rho)\qquad\text{(for $i=1,2$)}
\end{equation}
Indeed, any decomposition of $\hat\rho_1$ and $\hat\rho_2$ induces a decomposition of $\hat \rho$. So if $\ket{\psi}$ is part of a decomposition of $\hat\rho_1$ or $\hat\rho_2$ then it is immediately part of a decomposition of $\hat\rho$. Taking the supremum, it follows that 
\begin{equation}
    \sup\limits_{\ket{\psi}\in\operatorname{Supp}(\hat\rho)} \max_i \Delta_{\ket{\psi}}( H_i)\geq \sup\limits_{\ket{\psi}\in\operatorname{Supp}(\hat\rho_j)} \max_i \Delta_{\ket{\psi}}( H_i)
\end{equation}
for $j=1,2$. With all of this, we can finally write
\begin{subequations}
    \begin{align}
        F^\text{S}_{\hat\rho}&=\frac{\mathcal Q_{\hat \rho}(\hat H_\text{coll})}{\sup\limits_{\ket{\psi}\in\operatorname{Supp}(\hat\rho)} \max_i \Delta_{\ket{\psi}}( H_i)}\\
        &\leq \frac{t\mathcal Q_{\hat \rho_1}(\hat H_\text{coll})+(1-t)\mathcal Q_{\hat \rho_2}(\hat H_\text{coll})}{\sup\limits_{\ket{\psi}\in\operatorname{Supp}(\hat\rho)} \max_i \Delta_{\ket{\psi}}( H_i)}\\
        &\leq t\frac{\mathcal Q_{\hat \rho_1}(\hat H_\text{coll})}{\sup\limits_{\ket{\psi}\in\operatorname{Supp}(\hat\rho_1)} \max_i \Delta_{\ket{\psi}}( H_i)}+(1-t)\frac{\mathcal Q_{\hat \rho_2}(\hat H_\text{coll})}{\sup\limits_{\ket{\psi}\in\operatorname{Supp}(\hat\rho_2)} \max_i \Delta_{\ket{\psi}}( H_i)}\\
        &=tF^\text{S}_{\hat\rho_1}+(1-t)F^\text{S}_{\hat\rho_2}
    \end{align}
\end{subequations}
\end{proof}  

\section{Computation of extensions of $F$ to mixed states}\label{appendix: mixed computation}
In this section, we present the results from Eq.~(\ref{eq: mixed state computation}), but first, we precisely define the model. We assume the system consists of qubits, with $\mathcal{H} = \mathbb{C}^2$ and the collective Hamiltonian $\hat{H}_\text{coll} = \hat{Z}_1 + \cdots + \hat{Z}_n$. We define the following states
\begin{align}
    \ket{\psi}=\frac{1}{\sqrt{2}}(\ket{0\cdots0}+\ket{1\cdots1}) && \ket{\varphi}=\hat{\1}^{\otimes \frac{n}{2}}\hat X^{\otimes \frac{n}{2}}\ket{\psi}=\frac{1}{\sqrt{2}}(\ket{0\cdots01\cdots1}+\ket{1\cdots10\cdots0}),
\end{align}
where for all computations involving $\ket{\varphi}$ (which are relevant when considering $\sigma_\epsilon$), we assume that $n$ is even for simplicity. To perform the computations, we introduce a basis for the Hilbert space $(\mathbb{C}^2)^{\otimes n}$. For a binary string $s = s_1\cdots s_n$ with $s_i \in {0,1}$, we define
 \begin{equation}
     \ket{\psi_s^i}=\frac{1}{\sqrt{2}}(\ket{s}+i\ket{\overline{s}}),
 \end{equation}
 where $\overline{s}$ is the string obtained by flipping each bit of $s$: $\overline{s}_k = 1 - s_k$. Thus, we can rewrite the states as $\ket{\psi} = \ket{\psi_0^+}$ and $\ket{\varphi} = \ket{\psi_c^+}$, where $0$ represents the string of all zeros, and $c = 0\cdots01\cdots1$ is the string consisting of $n/2$ zeros followed by $n/2$ ones. Since $\ket{\psi_s^i} = i \ket{\psi_{\overline{s}}^i}$, we restrict our attention to strings that begin with a $0$. We define the set $\mathcal{S}$ as follows: $\mathcal S=\{0s_2\cdots s_n|s_k=0,1\}$. For $s,s'\in\mathcal S$ and $i,j=\pm1$ we can verify the following orthogonality relation
 \begin{equation}
     \braket{\psi_s^i}{\psi_{s'}^j}=\delta_{s,s'}\delta_{i,j}.
 \end{equation}
 Since $\abs{\mathcal S\times \{1,-1\}}=2^{n}$, the states $\ket{\psi_s^i}$ form a basis of $(\C^2)^{\otimes n}$. The quantum Fisher information appears in the definitions of $F^\text{S}$ and $F^\text{R}$, so we first compute $\mathcal Q_{\hat\rho_\epsilon}(\hat H_\text{coll})$ and $\mathcal Q_{\hat\sigma_\epsilon}(\hat H_\text{coll})$. We recall the expression of the quantum Fisher information for mixed states
 \begin{equation}
     \mathcal Q_{\hat \rho}(\hat A)=2\sum_{\lambda_i+\lambda_j>0}\frac{(\lambda_i-\lambda_j)^2}{\lambda_i+\lambda_j}\abs{\bra{i}\hat A \ket{j}}^2,
 \end{equation}
where the density matrix has been diagonalized as $\hat \rho=\sum_i \lambda_i\ketbra{i}$, and the sum runs over all pairs of indices $(i,j)$ satisfying $\lambda_i+\lambda_j>0$. Additionally, we observe that
\begin{equation}
    \bra{\psi_s^i}\hat Z_i\ket{\psi_{s'}^j}=(-1)^{s'_k}\braket{\psi_s^i}{\psi_{s'}^{-j}}=(-1)^{s_k}\delta_{s,s'}\delta_{i,-j},
\end{equation}
which leads to $\bra{\psi_s^i}\hat H_\text{coll}\ket{\psi_{s'}^j}=\delta_{s,s'}\delta_{i,-j}\sum_{k=1}^n (-1)^{s_k}$. With this, we can now compute $\mathcal Q_{\hat \rho_\epsilon}(\hat H_\text{coll})$ given the eigendecomposition of $\hat \rho_\epsilon$. Using the closure relation $\hat{\1}=\sum_{s\in\mathcal S,i=\pm1} \ketbra{\psi_s^i}$, we have
\begin{equation}
    \hat \rho_\epsilon=\left[(1-\epsilon)+\frac{\epsilon}{2^n}\right]\ketbra{\psi_0^+}+\frac{\epsilon}{2^n}\ketbra{\psi_0^-}+\frac{\epsilon}{2^n}\sum_{s\in \mathcal S\setminus\{0\}, i=\pm1}\ketbra{\psi_s^i}.
\end{equation}
Since $\hat \rho_\epsilon$ has only two distinct eigenvalues, and only pairs of eigenvectors with different eigenvalues contribute to the quantum Fisher information, we can proceed with the calculation as follows
\begin{subequations}
    \begin{align}
        \mathcal Q_{\hat \rho_\epsilon}(\hat H_\text{coll})&=4\sum_{(s,i)\in\mathcal S\times\{1,-1\}\setminus\{(0\cdots0,1)\}}\frac{[(1-\epsilon)+\epsilon/2^n-\epsilon/2^n]^2}{(1-\epsilon)+\epsilon/2^n+\epsilon/2^n}\abs{\bra{\psi_0^+}\hat H_\text{coll}\ket{\psi_s^i}}^2\\
        &=4\sum_{(s,i)\in\mathcal S\times\{1,-1\}\setminus\{(0\cdots0,1)\}}\frac{(1-\epsilon)^2}{(1-\epsilon)+\epsilon/2^{n-1}}\left(\delta_{0,s}\delta_{1,-i}\sum_{k=1}^n(-1)^{s_k}\right)^2\\
        &=4\frac{(1-\epsilon)^2}{(1-\epsilon)+\epsilon/2^{n-1}}n^2.
    \end{align}
\end{subequations}
We follow the same procedure to compute $\mathcal Q_{\hat \sigma_\epsilon}(\hat H_\text{coll})$. First, we diagonalize $\hat\sigma_\epsilon$, explicitly identifying the eigenvectors corresponding to the eigenvalue 0.
\begin{equation}
    \hat \sigma_\epsilon=(1-\epsilon)\ketbra{\psi_0^+}+\epsilon\ketbra{\psi_c^+}+ 0\sum_{(s,i)\in\mathcal S\times\{1,-1\}\setminus\{(0,+),(c,+)\}}\ketbra{\psi_s^i}.
\end{equation}
In this case, there are three distinct eigenvalues. We then have
\begin{subequations}
    \begin{align}
        \mathcal Q_{\hat \sigma_\epsilon}(\hat H_\text{coll})&=4\frac{(1-2\epsilon)^2}{1}\abs{\underbrace {\bra{\psi_0^+}\hat H_\text{coll}\ket{\psi_c^+}}_{0}}^2+4\frac{(1-\epsilon)^2}{1-\epsilon}\sum_{(s,i)\in\mathcal S\times\{1,-1\}\setminus\{(0,+),(c,+)\}} \abs{\underbrace{\bra{\psi_0^+}\hat H_\text{coll}\ket{\psi_s^i}}_{n\delta_{s,0}\delta_{i,-1}}}\notag\\
        &\qquad+4\frac{\epsilon^2}{\epsilon}\sum_{(s,i)\in\mathcal S\times\{1,-1\}\setminus\{(0,+),(c,+)\}}\abs{\underbrace{\bra{\psi_c^+}\hat H_\text{coll}\ket{\psi_s^i}}_{0}}\\
        &=4(1-\epsilon)n^2.
    \end{align}
\end{subequations}
Since we are working with qubits, the local variance of any pure state $\Delta_{\ket{\phi}}(\hat H_k)\leq 1$, which allows us to set $A=1$ in the definition of $F^\text{R}$. This leads to the following expressions
\begin{align}
    F_{\hat\rho_\epsilon}^\text{R}=\frac{(1-\epsilon)^2}{(1-\epsilon)+\epsilon/2^{n-1}}n^2 && F_{\hat\sigma_\epsilon}^\text{R}=(1-\epsilon)n^2.
\end{align}
For similar reasons, we have $\sup\limits_{\ket{\psi}\in\operatorname{Supp}(\hat\tau)} \max_i \Delta_{\ket{\psi}}(\hat Z_i)=1$ for $\hat \tau=\hat\rho_\epsilon,\hat \sigma_\epsilon$, as $\ket{\psi_0^+}$ is a state in the support with maximal variance.\\

Finally, we demonstrate that $F^\text{CR}_{\hat\sigma_\epsilon}=(1-\epsilon)n^2$. By definition, we have 
\begin{equation}
    F^\text{CR}_{\hat\sigma_\epsilon}=\inf_{\{p_i,\ket{\psi_i}\}} \sum_i p_i \frac{\Delta_{\ket{\psi_i}}(\hat H_\text{coll})}{\max_k\Delta_{\ket{\psi_i}}(\hat Z_k)}
\end{equation}
where the optimization is done over all decomposition of the form $\hat \sigma_\epsilon=\sum_i p_i\ketbra{\psi_i}$. As shown in Appendix \ref{appendix: support}, the states appearing in the decomposition must belong to $(\ker\hat \sigma_\epsilon)^\perp=\operatorname{Vect}(\ket{\psi_0^+},\ket{\psi_c^+})$. Since we have
\begin{align}
    \hat Z_k\ket{\psi_0^+}=\ket{\psi_0^-} && \hat Z_k\ket{\psi_c^+}=\left\{\begin{array}{cc} \ket{\psi_c^-}, & \text{if $k\leq n/2$}\\ -\ket{\psi_c^-}, & \text{if $k>n/2$}  \end{array}\right.
\end{align}
 we can compute the variance of $\hat Z_k$ on the state $\ket{\phi}=\alpha\ket{\psi_0^+}+\beta\ket{\psi_c^+}$. We have
 \begin{subequations}
     \begin{align}
         \bra{\phi}\hat Z_k\ket{\phi}&=\abs{\alpha}^2\bra{\psi_0^+}\hat Z_k\ket{\psi_0^+}+\abs{\beta}^2\bra{\psi_c^+}\hat Z_k\ket{\psi_c^+}+\alpha^*\beta\bra{\psi_0^+}\hat Z_k\ket{\psi_c^+}+\alpha\beta^*\bra{\psi_c^+}\hat Z_k\ket{\psi_0^+}\\
         &=\abs{\alpha}^2\bra{\psi_0^+}\ket{\psi_0^-}\pm\abs{\beta}^2\bra{\psi_c^+}\ket{\psi_c^-}+\alpha^*\beta\bra{\psi_0^+}\ket{\psi_c^-}\pm\alpha\beta^*\bra{\psi_c^+}\ket{\psi_0^-}\\
         &=0.
     \end{align} 
 \end{subequations}
And
 \begin{subequations}
     \begin{align}
         \bra{\phi}\hat Z_k^2\ket{\phi}&=\abs{\alpha}^2\bra{\psi_0^+}\hat Z_k^2\ket{\psi_0^+}+\abs{\beta}^2\bra{\psi_c^+}\hat Z_k^2\ket{\psi_c^+}+\alpha^*\beta\bra{\psi_0^+}\hat Z_k^2\ket{\psi_c^+}+\alpha\beta^*\bra{\psi_c^+}\hat Z_k^2\ket{\psi_0^+}\\
         &=\abs{\alpha}^2\bra{\psi_0^+}\ket{\psi_0^+}+\abs{\beta}^2\bra{\psi_c^+}\ket{\psi_c^+}+\alpha^*\beta\bra{\psi_0^+}\ket{\psi_c^+}+\alpha\beta^*\bra{\psi_c^+}\ket{\psi_0^+}\\
         &=\abs{\alpha}^2+\abs{\beta}^2=1.
     \end{align} 
 \end{subequations}
 So that $\Delta_{\ket{\phi}}(\hat Z_k)=\bra{\phi}\hat Z_k^2\ket{\phi}-\bra{\phi}\hat Z_k\ket{\phi}^2=1$, which means that all pure states in the support of $\hat \sigma_\epsilon$ have the same local variance $\Delta(\hat Z_k)$. From this, it follows that 
 \begin{equation}
     F^\text{CR}_{\hat\sigma_\epsilon}=\inf_{\{p_i,\ket{\psi_i}\}} \sum_i p_i \frac{\Delta_{\ket{\psi_i}}(\hat H_\text{coll})}{1}=\inf_{\{p_i,\ket{\psi_i}\}} \sum_i p_i \Delta_{\ket{\psi_i}}(\hat H_\text{coll})=\frac{1}{4}\mathcal Q_{\hat\sigma_\epsilon}(\hat H_\text{coll})
 \end{equation}
 where the quantum Fisher information is one-fourth of the convex roof of the variance. Using the previously established results allows us to conclude the proof.
\section{General $k$-entanglement}\label{appendix: k sep}
In this section, we prove the inequalities of Eq.~(\ref{eq: k sep general ineq}) under the assumptions formulated in Sec.~\ref{subsec: k sep general statement}.

\begin{proof}
 First, the convexity hypothesis on $F$, allows us to consider only pure states, as the extension to mixed states is immediate. If a pure state $\ket{\psi}$ can be decomposed as
\begin{equation}
    \ket{\psi}=\ket{\phi_1}\otimes \cdots\otimes \ket{\phi_l}
\end{equation}
where $\ket{\phi_i}$ is a pure state in $\mathcal H^{\otimes k_i}\cap\Gamma_{k_i}$, the sub-additivity of $F$ shows that
\begin{equation}
    F(\ketbra{\psi})\leq F(\ketbra{\psi_1}+\cdots+F(\ketbra{\psi_l})\leq f(k_1)+\cdots +f(k_l)
\end{equation}
We then need to show that for any finite family $k_1,\dots,k_r$ of positive real numbers, such that $k_i\leq k$ for all $i$ and $k_1+\cdots+k_r=n$, we have 
\begin{equation}
    f(k_1)+\cdots+f(k_r)\leq \left\lfloor \frac{n}{k}\right\rfloor f(k)+f\left(n-\left\lfloor \frac{n}{k}\right\rfloor k\right)\leq \frac{n}{k}f(k)
\end{equation}

For the first inequality, we reason by induction on $r\in\N^*$. For $r=1$, if $k_1\leq k$ is such that $k_1=n$, then $k\geq n$, and the inequality is simply verified, provided $f(0)=0$. For $r>1$, we consider the set
\begin{equation}
    C^n_r=\Big\{(x_1,\dots,x_r)\Big|x_1+\cdots x_r=n\text{ and } 0\leq x_i\leq k\Big\}
\end{equation}
which is a compact convex set, and the function $g_r$ defined on $C^n_r$ by $g_r(x_1,\dots,x_r)=f(x_1)+\cdots+f(x_r)$. Since $f$ is convex, it follows that $g_r$ is also convex. By the Bauer principle (a convex function on a compact convex set attains a maximum on one of its extreme points), we know that $g$ attains a maximum on an extreme point of $C_r^n$. Extreme points of $C^n_r$, are given by points $(k_1,\dots,k_r)$ where at least one $k_i$ is equal to $0$ or $k$. 
\begin{itemize}
    \item If $k_i=0$ then the maximum of $g_r$ on $C^n_r$ is the same as the one of $g_{r-1}$ on $C^n_{r-1}$
    \item If $k_i=k$ then the maximum of $g_r$ on $C^n_r$ is $\max\limits_{C^{n-k}_{r-1}} g_{r-1}+f(k)$
\end{itemize}
In both cases, we conclude by induction.\\

For the second inequality, we notice that it is satisfied if and only if
\begin{equation}
    \frac{f(k)}{k}\geq \frac{f\left(n-\left\lfloor \frac{n}{k}\right\rfloor k\right)}{n-\left\lfloor \frac{n}{k}\right\rfloor k}
\end{equation}
Recall a useful property of convex function, usually called the three-secant inequality: if $x<y<z$ then
\begin{equation}
    \frac{f(y)-f(x)}{y-x}\leq \frac{f(z)-f(x)}{z-x}\leq \frac{f(z)-f(y)}{z-y}
\end{equation}
Using the left inequality for $x=0$, (with $f(0)=0$), $y=n-\left\lfloor \frac{n}{k}\right\rfloor k$ and $z=k$, directly yield the result.
\end{proof}

\section{Popoviciu's inequality}\label{appendix popoviciu's}
In this section, we prove the following: for a Hilbert space $\mathcal H$, a Hamiltonian $\hat H$ with largest and smallest eigenvalues $h_{\max}$, $h_{\min}$, we have
\begin{equation}
    \Delta \hat H\leq \frac{(h_{\max}-h_{\min})^2}{4}
\end{equation}
with equality for states $\ket{\psi}=\frac{1}{\sqrt{2}}\left[\ket{\psi_{\min}}^{\otimes n}+\ket{\psi_{\max}}^{\otimes n}\right]$ where $\ket{\psi_{\min}}$ (resp. $\ket{\psi_{\max}}$) is any eigenstate with eigenvalue $h_{\min}$ (resp. $h_{\max}$). This statement is a direct consequence of a classical probability result known as Popoviciu's inequality \cite{bhatia_better_2000}. In the following, we state and prove it.

\begin{thm}(Popoviciu's inequality)
    For any real random variable $X$ that is bounded (almost surely), if we set $m=\inf X$, $M=\sup M$. (More precisely, we set $m=\sup\{x\in\R|\mathbb P(X\geq x)=1\}$ and $M=\inf\{x\in\R|\mathbb P(X\leq x)=1\}$, which is rigorously not the same thing). Then we have
    \begin{equation}
        \mathbb V(X)\leq \frac{(M-m)^2}{4}
    \end{equation}
    Furthermore, we have equality exactly when $\mathbb P(X=m)=\mathbb P(X=M)=1/2$.
\end{thm}

\begin{proof}
    Since we have $\mathbb P(m\leq X\leq M)=1$, the random variable $(X-m)(M-X)$ is non-negative almost surely. It thus has a non-negative expectation value
    \begin{equation}
        \mathbb E\big[(X-m)(M-X)\big]\geq 0
    \end{equation}
    Expanding the expectation value and subtracting $\mathbb E[X]^2$ on both sides gives
    \begin{subequations}
        \begin{align}
            &\mathbb E\big[(X-m)(M-X)\big]\geq 0\\
            \Rightarrow& \mathbb E\big[MX-X^2-mM+mX\big]\geq 0\\
            \Rightarrow& M\mathbb E(X)-\mathbb E(X^2)-mM+m\mathbb E(X)\geq 0\\
            \Rightarrow& \mathbb E(X^2)-\mathbb E(X)^2\leq M\mathbb E(X)-mM+m\mathbb E(X)-\mathbb E(X)^2\\
            \Rightarrow& \mathbb V(X)\leq (\mathbb E(X)-m)(M-\mathbb E(X))
        \end{align}
    \end{subequations}
    Finally we can use the classic inequality $ab\leq \frac{(a+b)^2}{4}$ for $a=\mathbb E(x)-m$ and $b=M-\mathbb E(X)$ to get
    \begin{equation}
        \mathbb V(X)\leq \frac{(\mathbb E(X)-m+M-\mathbb E(X))^2}{4}=\frac{(M-m)^2}{4}
    \end{equation}
    To understand the case of equality, we simply have to look at the two steps where we put an inequality. First, we said that the random variable $(X-m)(M-X)$ is non-negative thus, its expectation value has to be non-negative. However, it is a well know result, that for a non-negative random variable $Y$ if $\mathbb E(Y)=0$ then $\mathbb P(Y=0)=1$. In our case, this means that
    \begin{equation}
        \mathbb P[(X-m)(M-X)=0]=1
    \end{equation}
    As a product is zero if and only if one of the factors is zero, we know that almost surely either $X=m$ or $X=M$. Denoting $p=\mathbb P[X=m]$ we get that $\mathbb P[X=M]=1-p$.\\
    Now, we can consider the second inequality we used: $ab\leq (a+b)^2/4$. It is easy to see that we have equality only when $a=b$, which imposes
    \begin{equation}
        \mathbb E(X)-m=M-\mathbb E(X)\Rightarrow \mathbb E(X)=\frac{m+M}{2}
    \end{equation}
    On the other hand, we can compute the expectation value since we know the exact distribution of $X$
    \begin{equation}
        \mathbb E(x)=m\mathbb P[x=m]+M\mathbb P[X=M]=pm+(1-p)M
    \end{equation}
    Asking that both expressions are the same, from which we can solve $p=1/2$.
\end{proof}

\end{document}